\documentclass[journal,12pt,onecolumn,draftclsnofoot,]{IEEEtran}
\usepackage[T1]{fontenc}
\usepackage{tikz}
\usepackage{etoolbox}
\usepackage{enumitem}
\newcommand*{\circled}[1]{\lower.7ex\hbox{\tikz\draw (0pt, 0pt)%
    circle (.5em) node {\makebox[1em][c]{\small #1}};}}
\robustify{\circled}
\usepackage{amsmath}
\usepackage{amsthm}
\usepackage{amssymb} % 为了使用黑色方块符号
\usepackage{hyperref}
\usepackage{mathrsfs}
\usepackage{bbm}
\usepackage{amsfonts}
\usepackage{graphicx,cite,epsfig,amssymb,amsmath}
\usepackage{color,xcolor}
\usepackage{bm}
\definecolor{red}{rgb}{1.00, 0.00, 0.00}  
\usepackage{pifont}
\usepackage{stmaryrd}
\usepackage{setspace}
\usepackage{subfigure}
\usepackage{cite}
\usepackage{array}
\usepackage{float}
\usepackage{multirow}
\usepackage{algorithm,algpseudocode}
\usepackage{epstopdf}
\usepackage{mathtools}
\usepackage{amssymb}
\usepackage{diagbox}
\usepackage{booktabs}
\usepackage{makecell}

\usepackage{tcolorbox}

\usepackage{makecell}
\usepackage{xcolor}
\usepackage{pifont}
\usepackage{dsfont}
\makeatletter

\newtheoremstyle{examplestyle}
  {}% Space above
  {}% Space below
  {}% Body font
  {}% Indent amount
  {\itshape}% Theorem head font
  {}% Punctuation after theorem head
  {.5em}% Space after theorem head
  {\thmname{#1}\thmnumber{ #2}\thmnote{ (#3)}}% Theorem head spec
\theoremstyle{examplestyle}
\newtheorem{example}{Example}

\newcommand{\Rmnum}[1]{\expandafter\@slowromancap\romannumeral #1@}

\makeatother
\usepackage{pifont} 
\usepackage{epstopdf}
\usepackage{mathtools}
\usepackage{diagbox}
\usepackage{booktabs}
\usepackage{multirow}
\usepackage{booktabs}
\usepackage{makecell}
\usepackage{cancel}
\usepackage{float} %设置图片浮动位置的宏包
\usepackage{subfigure} %插入多图时用子图显示的宏包
\usepackage{colortbl}
\newtheorem{remark}{Remark}
\usepackage{setspace}
\newtheorem{proposition}{Proposition}
\newtheorem{corollary}{Corollary}

\definecolor{SP}{RGB}{30,160,156}

\usepackage{setspace}
\usepackage{tikz}
\renewenvironment{proof}{{\itshape Proof.}}{}
\begin{document}
\normalsize
\title{Reliable LLM-Based Edge-Cloud-Expert Cascades for Telecom Knowledge Systems}

\author{Qiushuo Hou \IEEEmembership{Graduate Student Member, IEEE}, Sangwoo Park \IEEEmembership{Member, IEEE}, Matteo Zecchin \IEEEmembership{Member, IEEE}, Yunlong Cai \IEEEmembership{Senior Member, IEEE}, Guanding Yu \IEEEmembership{Senior Member, IEEE}, Osvaldo Simeone \IEEEmembership{Fellow, IEEE}, and Tommaso Melodia \IEEEmembership{Fellow, IEEE}
\thanks{

Q. Hou, Y. Cai, and G. Yu are with the College of Information Science and Electronic Engineering, Zhejiang University, Hangzhou 310027, China (e-mail: \{qshou, ylcai, yuguanding\}@zju.edu.cn).

S. Park and M. Zecchin are with the King’s Communications, Learning \& Information Processing (KCLIP) lab within the Centre for Intelligent Information Processing Systems (CIIPS), Department of Engineering, King’s College London, London WC2R 2LS, U.K. (e-mail: \{sangwoo.park, matteo.1.zecchin\}@kcl.ac.uk). 

{\color{black} T. Melodia and O. Simeone are with the Intelligent Networked Systems Institute (INSI) at Northeastern University (TM is in Boston, MA, USA and OS in London, UK).}

The work of M. Zecchin and O. Simeone is supported by an Open Fellowship of the EPSRC (EP/W024101/1). The work of O. Simeone is also supported by the EPSRC project (EP/X011852/1). The work of T. Melodia is based on material supported in part by the National Science Foundation under award CNS-2112471.
}}

\maketitle
\vspace{-1.5cm}
\begin{abstract}
Large language models (LLMs) are emerging as key enablers of automation in domains such as telecommunications, assisting with tasks including troubleshooting, standards interpretation, and network optimization. However, their deployment in practice must balance inference cost, latency, and reliability. In this work, we study an edge-cloud-expert cascaded LLM-based {\color{black}knowledge} system {\color{black}that supports decision-making through a question-and-answer pipeline. In it,} an efficient edge model handles routine queries, a more capable cloud model addresses complex cases, and human experts are involved only when necessary. We define a misalignment-cost constrained optimization problem, aiming to minimize average processing cost, while guaranteeing alignment of automated {\color{black}answers} with expert judgments. We propose a statistically rigorous threshold selection method based on {multiple hypothesis testing (MHT)} for a query processing mechanism based on knowledge and confidence tests. The approach provides finite-sample guarantees on misalignment risk. Experiments on the TeleQnA dataset -- a telecom-specific benchmark -- demonstrate that the proposed method achieves superior cost-efficiency compared to conventional cascaded baselines, while ensuring reliability at prescribed confidence levels. 
\end{abstract}


\section{Introduction} \label{sec:intro}

\subsection{Motivation}
Large language models (LLMs) are increasingly being integrated into telecommunications to automate complex tasks and support network management. They are expected to play a central role in {\color{black} supporting decision process, as well as in closed-loop monitoring and control pipelines} \cite{LLM_4_comm_1, LLM_4_comm_2}. {\color{black}In this context, an important use case is given by LLM-based knowledge system for decision support via question-and-answer pipelines \cite{lewis2020retrieval, zhu2025large, papachristou2025leveraging}. However}, evaluations of open-source LLMs on telecom benchmarks reveal a lack of reliability guarantees \cite{LLM_evulate}. This limitation is critical, as misconfigurations or misdiagnoses by automated assistants can cause severe service disruptions and financial losses \cite{LLM_reliability_1, LLM_reliability_2}. In addition, deploying state-of-the-art LLMs such as GPT-5 incurs high computational cost and latency, particularly on resource-constrained edge devices \cite{LLM_latency_1, LLM_latency_2}.

These challenges motivate the design of hybrid edge–cloud–expert pipelines (Fig. \ref{fig:cascaded_model}), where lightweight edge models address routine queries, complex queries are escalated to a cloud LLM, and only the most uncertain cases are referred to human experts. Such cascading seeks to balance efficiency, accuracy, and cost: the edge offers low latency and privacy, the cloud delivers higher accuracy at greater expense, and human experts ensure correctness when automation is insufficient. While routing all queries to experts would maximize reliability, this approach is neither scalable nor economically viable. Instead, effective systems must guarantee strong alignment between LLM-generated outputs and expert-verified solutions, escalating only when necessary to maintain a target reliability level.

To meet this goal, we propose a query processing framework for cascaded edge–cloud–expert LLM{\color{black}-based knowledge} systems that provides formal reliability guarantees. The framework employs knowledge and confidence tests at both the edge and cloud stages to decide when escalation is warranted. By formulating the problem as a misalignment-constrained optimization, the method minimizes average processing cost while ensuring that automated decisions align with expert judgments at a user-specified confidence level.

\subsection{Related Work}
\noindent \emph{{Cascaded decision-making and routing in AI systems}}: There is a rich history of cascaded inference frameworks that trade off accuracy for efficiency. In distributed AI, model cascades employ a sequence of models of increasing capacity, where an inexpensive model first attempts the task and uncertain cases are passed to a stronger model. For example, FrugalGPT \cite{chen2023frugalgpt} dynamically routes queries between smaller and larger LLMs to reduce API costs while preserving answer quality. Similarly, SpecInfer \cite{miao2024specinfer} employs a speculative execution approach where a small draft model generates initial responses that are then verified and refined by a larger target model, achieving significant speedup in text generation tasks. Another notable example is the cascade framework proposed in LLMCascade \cite{LLM_cascade}, which uses confidence-based routing across multiple LLM APIs (e.g., from GPT-3.5 to GPT-4) based on query complexity assessment. Beyond model-only cascades, researchers have explored human–AI collaboration in decision systems. In particular, learning to defer frameworks allows an AI model to abstain and delegate a query to a human expert when the model’s confidence is low \cite{to_cascade_or_not, cambridge_paper}. Such methods train a classifier along with a rejector to decide between giving an answer or deferring to an expert, aiming to minimize overall risk. 

\noindent \emph{{LLM uncertainty estimation and reliable inference:}}  With LLMs being deployed in high-stakes settings, understanding and calibrating their uncertainty has become an important topic of research. Modern LLMs often exhibit overconfidence and can produce fluent but incorrect answers \cite{ensemble_survey, deepmind, test_time_scaling_1}. To mitigate this, various techniques have been proposed to estimate an LLM’s confidence in its responses. One practical approach is ensemble or prompt aggregation. For instance, reference \cite{ensemble} leverages multiple prompt variations as an implicit ensemble to quantify uncertainty and improve calibration without retraining the model. Reference \cite{auxiliary_model_uncertainty} uses an auxiliary model to analyze disagreement patterns among multiple responses to diagnose uncertainty sources. In addition, for white-box models, researchers have explored Bayesian methods such as Monte Carlo dropout, deep ensembles, and evidential deep learning to provide probabilistic confidence estimates in a single forward pass \cite{bayesian_prompt,cambridge_paper,zellinger2025rational}. 

Alongside uncertainty quantification, there is growing interest in routing queries based on confidence: a system can decide to answer immediately or route the query to a more expert tier if the LLM’s predicted confidence is below a threshold \cite{to_cascade_or_not, cambridge_paper}. 

\noindent \emph{{Statistical reliability for black-box models:}}  To make LLM decisions reliable, recent works seek formal statistical guarantees. The learn-then-test (LTT) framework reframes model calibration as a risk-controlling hypothesis test, enabling explicit finite-sample guarantees on error rates \cite{ltt}. By using techniques from multiple hypothesis testing (MHT), LTT can adjust a model’s prediction set or confidence threshold so that the probability of failure (e.g., a misalignment with the ground truth or expert decision) stays below a desired level \cite{ltt,ltt_adaptive,ltt_quantile, amir}. 

In the context of LLM deployment, specialized methods have also emerged to ensure model outputs meet reliability criteria. For example, the work \cite{to_cascade_or_not} proposes a ``trust or escalate'' mechanism in which an auxiliary LLM-based judge evaluates the main LLM’s answer and decides whether to trust it or defer to a human, providing provable guarantees on agreement with human experts. %Such approaches combine uncertainty estimation with decision policies to guarantee that high-risk queries are identified and handled appropriately.

%Overall, existing cascade and routing frameworks do not typically incorporate formal statistical reliability guarantees, and conversely, risk-controlled methods have not been explored in multi-tier LLM deployment scenarios under cost constraints.
{\color{black}{Overall, existing cascade and routing frameworks such as FrugalGPT \cite{chen2023frugalgpt} and LLMCascade \cite{LLM_cascade} do not incorporate formal statistical reliability guarantees. Learning-to-defer framework addresses human-AI collaboration, but requires joint training of the classifier and deferral policy on a dataset of expert decisions \cite{liu2024mitigating}. Finally, risk-controlled methods such as LTT \cite{ltt} have not been explored in multi-tier LLM deployment scenarios under cost constraints. Our work bridges these gaps by equipping a cascaded LLM routing system with post-hoc, training-free finite-sample reliability guarantees.}}

\subsection{Main Contributions}
 In this work, we bridge the gaps identified above by proposing a cascaded LLM system for question-answering {\color{black}pipelines} that is not only cost-aware but also equipped with rigorous reliability assurances. The main contributions of this paper are summarized as follows:
\begin{itemize}
\item We formulate an \emph{edge–cloud–expert LLM cascading system}, with a focus on the use case of automatic question answering for telecommunications. This framework, illustrated in Fig. \ref{fig:cascaded_model}, integrates an edge LLM, a more powerful cloud LLM, and human expertise in a unified decision pipeline. For this system, we introduce a query processing mechanism that leverages knowledge tests based on epistemic uncertainty measures and confidence tests based on the model output. Both white-box and black-box implementations are considered.

\item We develop a \emph{misalignment-cost constrained optimization framework} that balances the trade-off between decision quality and inference cost. The misalignment metric measures the discrepancy between model decisions and expert decisions, while the cost accounts for the computational and human resources required to address a query at the edge, cloud, or by a human expert. We focus on minimizing the average cost subject to an alignment reliability constraint.

\item We propose a \emph{statistically reliable threshold selection method} based on MHT that chooses uncertainty and confidence thresholds used in the cascading decision process while guaranteeing a target upper bound on the misalignment rate. This approach provides finite-sample assurances that the cascading system meets predefined reliability levels with high confidence.

\item We perform \emph{extensive experimental validation} on the TeleQnA dataset 
 -- a recent telecom QA benchmark \cite{teleQnA} -- under realistic operational scenarios. The results demonstrate that our cascaded system achieves superior cost-efficiency compared to single-model baselines, while maintaining the reliability requirements of mission-critical telecom applications. The experiments validate our theoretical framework under realistic operational scenarios and provide a detailed analysis of key system parameters, including calibration dataset size, violation upper bounds, and grid resolution. To evaluate system adaptability, we further investigate a \emph{reasoning-enhanced} cloud deployment scenario where the cloud model is upgraded to Qwen3-4B with enhanced reasoning capabilities controlled by the number of thinking tokens allocation.

 \subsection{Organization}
 The rest of the paper is organized as follows. Sec. \ref{problem formulation} introduces the problem of reliable decision-making for the cascaded LLM decision system. Sec. \ref{background} reviews practical solutions for the formulated misalignment-cost constrained optimization problem. The proposed threshold selection method is presented in Sec. \ref{proposed method}. Sec. \ref{simulation} describes the simulation setup and illustrates the simulation results. Finally, Sec. \ref{conclusion} concludes the paper.

 % The results demonstrate that our cascaded system achieves superior cost-efficiency compared to single-model baselines, while satisfying the reliability requirements of mission-critical telecom settings. We report ablation studies and case analyses highlighting how the system adaptively routes queries and the impact of different confidence threshold settings. {\color{red} this part needs to be updated}
\end{itemize}

\section{Setting and Problem Definition}\label{problem formulation}
In this section, we first introduce the cascaded LLM decision system in Fig.~\ref{fig:cascaded_model}, and then we define the problem of reliable decision-making to be studied in this paper.

\subsection{Setting}
As illustrated in Fig. \ref{fig:cascaded_model}, we consider a three-tiered edge-cloud-expert cascaded LLM decision system {\color{black}for question-answering that controls} the level of alignment between {\color{black}its answers}  and expert (human) judgement {\color{black}with minimal} cost {\color{black}in terms of system resources}. We specifically envision a deployment in which a cheaper{\color{black}, i.e., lower-cost,} edge model should address most queries, while relying on the more {\color{black}costly} cloud model and on human intervention only when necessary to ensure a reliable response \cite{cambridge_paper}. While the applicability of the methodology developed in this study is broader, we
adopt as a running example an automatic expert system for telecommunications networks, in which users' technical queries are handled by the edge, the cloud, or the human expert depending on their difficulty level \cite{LLM_cascade, to_cascade_or_not, cambridge_paper}. 

\begin{figure*}[htbp]
    \centering
    \includegraphics[width=16cm]{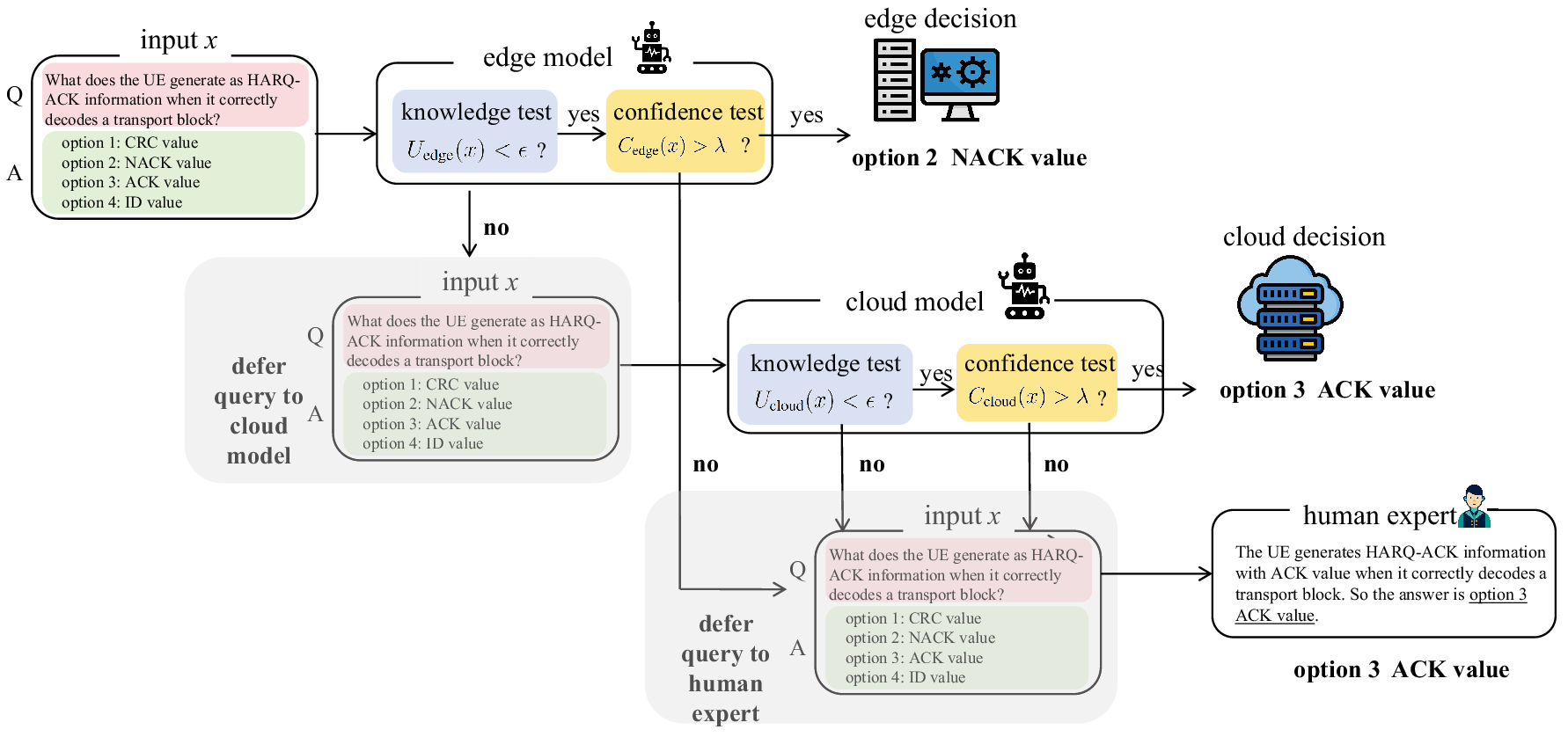}
    \caption{Cascaded edge-cloud-human system: The query is processed by the edge model $M_{\text{edge}}$ if the edge model's epistemic uncertainty $U_{\text{edge}}(x)$ remains within the acceptable level $\epsilon$, while the confidence $C_{\text{edge}}(x)$ exceeds a threshold $\lambda$, i.e., $U_{\text{edge}}(x) < \epsilon$ and $C_{\text{edge}}(x) > \lambda$. Thus, the edge decision $M_{\text{edge}}(x)$ is produced only if the edge model is sufficiently knowledgeable and confident. When the edge epistemic uncertainty condition is not met, and thus the edge model does not have sufficient knowledge to address the query, the input $x$ is forwarded to the cloud model $M_{\text{cloud}}$. Similar knowledge and confidence tests are carried out for the cloud model based on epistemic uncertainty measure $U_{\text{cloud}}(x)$ and confidence measure $C_{\text{cloud}}(x)$. If the cloud model passes the test, i.e.,  $U_{\text{cloud}}(x) < \epsilon$ and $C_{\text{cloud}}(x) > \lambda$, the cloud decision $M_{\text{cloud}}(x)$ is returned, otherwise, the input $x$ is deferred to a human expert.}
    \label{fig:cascaded_model}
    % \vspace{2cm}
\end{figure*}

As illustrated in Fig. \ref{fig:cascaded_model}, and further detailed in Sec. \ref{sec:query}, we adopt a modified version of the cascading framework introduced in \cite{cambridge_paper}, in which the decisions to process a query using the edge or cloud models are based on  \emph{knowledge} and \emph{confidence} tests that leverage the epistemic uncertainty and confidence scores for the two models. 

To elaborate, let $x \in \mathcal{X}$ be an input query. For every input $x$, the models at the edge and at the cloud, denoted respectively as $M_\text{edge}$ and $M_\text{cloud}$, produce the following scores:
%have associated \emph{epistemic uncertainty scores} $U_{\text{edge}}(x) \in [0, 1]$ and $U_{\text{cloud}}(x) \in [0, 1]$, as well as \emph{confidence scores} $C_{\text{edge}}(x) \in [0, 1]$ and $C_{\text{cloud}}(x) \in [0, 1]$:
\begin{itemize}
    \item \emph{Epistemic uncertainty scores}: The {epistemic uncertainty scores} $U_{\text{edge}}(x)\in[0,1]$ and $U_{\text{cloud}}(x)\\\in[0,1]$  reflect the respective levels of \emph{knowledge} that edge and cloud have on the query encoded by input $x$.  %level of ignorance, and thus conversely of expertise or knowledge, of the edge and cloud models, respectively, for the query encoded by input $x$. 
    %%a higher ignorance, i.e., 
    Larger scores indicate higher epistemic uncertainty and thus lower knowledge levels. As will be further discussed in Sec. \ref{obtain_u_and_c}, epistemic uncertainty scores are typically obtained via ensembling methods that evaluate the \emph{disagreement}  between decisions made by the model for input $x$ under different conditions. A higher disagreement expresses a larger epistemic uncertainty. Specifically, one can construct ensembles of the same model with different weights when \emph{white-box} access to the model is available \cite{cambridge_paper, bayesian_prompt}; or leverage variations of the prompt template to query a single, possibly \emph{black-box}, model  \cite{deepmind, ensemble, ensemble_survey}. 

    \item \emph{Confidence scores}: The confidence scores $C_{\text{edge}}(x)\in[0,1]$ and $C_{\text{cloud}}(x)\in[0,1]$ measure the predictive confidence levels associated with the edge model and cloud model, respectively. Larger scores indicate lower prediction uncertainty.  The confidence score can be obtained directly from the predictive distribution assigned by the model to the output tokens when one has \emph{white-box} access to the model  \cite{ensemble_survey}; or  through self-confidence measures when model access is of \emph{black-box}  nature \cite{ensemble_survey,test_time_scaling_1}.
\end{itemize}

{\color{black}{The two scores serve complementary roles in the routing decision. The epistemic uncertainty score assesses whether the model has sufficient domain knowledge to address the query. In contrast, the confidence score measures how certain the model is about its specific answer on the basis of that knowledge, while accounting also for potential inherent ambiguities in the given query. Both scores must be controlled in order for the system to trust the model's output. In fact, a model may be knowledgeable about a topic, yet still uncertain about a particular ill-posed or ambiguous query, motivating the confidence test. Conversely, it may produce a confident-sounding answer, while lacking genuine domain coverage, motivating the knowledge test.}}

% This uncertainty is generally contributed by both the inherent aleatoric uncertainty of the prediction, and by the epistemic uncertainty arising from the knowledge limitation of the edge model. 
% Note that the scores $U_{\text{edge}}(x)$ and $C_{\text{edge}}(x)$ are only estimates of the epistemic uncertainty and the true prediction confidence, respectively, and that the corresponding discrepancies between true and estimated measures can be measured by some calibration metrics [\textcolor{blue}{ref}]. In a similar way, the cloud model produces the epistemic uncertainty score $U_{\text{cloud}}(x)$ and confidence score $C_{\text{cloud}}(x)$.

\subsection{Query Processing}\label{sec:query}

In a manner inspired by \cite{cambridge_paper}, a query $x$ is processed by the edge-cloud-expert system through a threshold-based routing mechanism that leverages knowledge and confidence tests. Specifically, as shown in Fig. \ref{fig:cascaded_model}, the edge model first processes the query $x$ by computing its epistemic uncertainty score $U_\text{edge}(x)$ and its confidence score $C_\text{edge}(x)$. The edge decision $M_\text{edge}(x)$ is produced as the final output of the system if the following two conditions are satisfied: 
\begin{itemize}
    \item \emph{Edge knowledge test}: The edge decision passes the {knowledge test}, in the sense that the epistemic uncertainty score $U_{\text{edge}}(x)$ does not exceed a pre-defined level $\epsilon$:
    \begin{align} \label{eq:kt_edge}
        U_{\text{edge}}(x) < \epsilon.
    \end{align}
    \item \emph{Edge confidence test}: The edge decision also passes {confidence test}, in the sense that the confidence level $C_{\text{edge}}(x)$ is no smaller than a pre-defined level $\lambda$:
    \begin{align} \label{eq:ct_edge}
        C_{\text{edge}}(x) > \lambda.
    \end{align}
\end{itemize}
Thus, the edge decision $M_{\text{edge}}(x)$ is returned only if the edge model is sufficiently knowledgeable (as per (\ref{eq:kt_edge}))  and confident (as per (\ref{eq:ct_edge})). 

If the knowledge test is not passed by the edge model, the input $x$ is escalated to the cloud model, under the assumption that the cloud model may possess broader knowledge and thus be better equipped to handle the query.

Conversely, if the knowledge test is passed but the confidence test fails, this indicates that the edge model is familiar with the query but uncertain about producing a reliable answer. In this case, the routing strategy in Fig. \ref{fig:cascaded_model} forwards the query directly to the human expert, relying on the expert’s ability to disambiguate the query and provide a dependable response.
% { Otherwise -- either knowledge test is not passed (\emph{lack of edge knowledge case}) or  knowledge test is passed but not the confidence test (\emph{query per se being too difficult case}) --  the input $x$ is forwarded to the cloud (former case) or to the human expert (latter case). }

%\subsection{Query Processing at the Cloud}\label{sec:query}

Upon receiving a deferred input $x$ from the edge, the cloud model evaluates its own scores $U_\text{cloud}(x)$ and $C_\text{cloud}(x)$. The cloud decision $M_{\text{cloud}}(x)$  is produced as the final system output as long as it passes both the \emph{cloud knowledge test} and the \emph{cloud confidence test}, i.e., if 
\begin{align} \label{eq:t_cloud}
    \text{$U_{\text{cloud}}(x) < \epsilon$ and $C_{\text{cloud}}(x) > \lambda$.}
\end{align}
If any of the test fails, the input $x$ is forwarded to a human expert. 

{\color{black}Denote as $\phi=[\epsilon, \lambda]$ the thresholds parameterizing the routing decisions (\ref{eq:kt_edge})--(\ref{eq:t_cloud}).} Overall, using (\ref{eq:kt_edge})--(\ref{eq:t_cloud}), the output of the  edge-cloud-expert decision process is given by \begin{equation}\label{cas}
f_{\phi}(x) = 
\begin{cases}
M_{\text{edge}}(x) & \text{if } C_{\phi,\text{edge}}(x)=1 \\
M_{\text{cloud}}(x) & \text{if } C_{\phi,\text{cloud}}(x)=1 \\
y & \text{otherwise},
\end{cases}
\end{equation}
where \begin{equation} C_{\phi,\text{edge}}(x) = \mathds{1}\{  U_{\text{edge}}(x) < \epsilon  \ \text{and}\  C_{\text{edge}}(x) > \lambda\},\end{equation} with $\mathds{1}(\text{true})=1$ and $\mathds{1}(\text{false})=0$, equals 1 if the system decision is produced at the edge, and   
\begin{equation} 
C_{\phi,\text{cloud}}(x) = \mathds{1}\{ U_{\text{edge}}(x) \geq \epsilon\  \text{and}\  U_{\text{cloud}}(x) < \epsilon\  \text{and}\  C_{\text{cloud}}(x) > \lambda\}
\end{equation} 
% \begin{multline}
%     C_{\phi,\text{cloud}}(x) = \mathds{1}\{ U_{\text{edge}}(x) \geq \epsilon\  \text{and}\  U_{\text{cloud}}(x) < \epsilon\ \text{and}\  C_{\text{cloud}}(x) > \lambda\}.
% \end{multline}
equals $1$ if the decision is produced at the cloud.
% which depends on the parameter $\phi=[\epsilon,\lambda]$ that collects the two thresholds $\epsilon$ and $\lambda$ used in the knowledge and confidence tests in (\ref{eq:kt_edge})--(\ref{eq:t_cloud}), respectively, and $y$ represents the ground-truth answer produced by the human expert. { To simplify the notation in the following, we will denote as $r_\phi(x) \in \{ \text{e}, \text{c}, \text{h}\}$ the chosen routing for the input $x$ with which $f_\phi(x)$ in (\ref{cas}) can be rewritten as 
% \begin{align}
%     f_\phi(x) = M_\text{edge}(x) \cdot \mathbbm{1}\{r_\phi(x)=\text{e}\} + M_\text{cloud}(x)\cdot \mathbbm{1}\{r_\phi(x)=\text{c}\} + y \cdot \mathbbm{1}\{r_\phi(x)=\text{h}\}.
% \end{align}
% }

\subsection{Problem Definition}
For an input-output pair $(x,y)$, where $y$ represents human expert judgment for the input $x$, the system performance is measured in terms of misalignment and cost, which are defined as follows.

\subsubsection{Misalignment loss} The misalignment loss equals $1$ if the cascaded system fails to reflect the human expert decision $y$, while returning $0$ otherwise. Mathematically, the misalignment loss is defined as 
\begin{equation}\label{misalignmnet}
    \mathcal{A}(\phi|x,y) = \mathds{1}\{f_{\phi}(x)\neq y \}.
\end{equation}

\subsubsection{System cost} Denoting as $L_{\text{edge}}$, $L_{\text{cloud}}$, and $L_{\text{human}}$ the non-negative scalar costs associated with the evaluation of a decision at the edge, cloud, and expert, respectively, the system cost is defined as 
% {\color{red} please rewrite the cost (commented out below) by using (5)-(6)}
% \begin{equation}\label{cost}
%     \mathcal{L}(\phi|x) = { 
%         L_{\text{edge}} \cdot \mathds{1}\{r_{\phi}(x)= \text{e}\} +L_{\text{cloud}}\cdot \mathds{1}\{r_{\phi}(x)= \text{c}\} +L_{\text{human}}\cdot \mathds{1}\{r_{\phi}(x)= \text{h}\}.}
% \end{equation}
% \begin{equation}\label{cost}
%     \begin{aligned}
%         &\mathcal{L}(\phi|x) = C_{\phi,\text{edge}}(x) \cdot L_{\text{edge}} + C_{\phi,\text{cloud}}(x) \cdot L_{\text{cloud}}\\ &+ \mathds{1}\{(U_{\text{edge}}(x) < \epsilon \land C_{\text{edge}}(x) \leq \lambda) \lor (U_{\text{edge}}(x) \geq \epsilon \land (U_{\text{cloud}}(x) \geq \epsilon \lor C_{\text{cloud}}(x) \leq \lambda))\} \cdot L_{\text{human}}.
%     \end{aligned}
% \end{equation}
    \begin{equation}\label{cost}
    % \begin{aligned}
        \mathcal{L}(\phi|x) = C_{\phi,\text{edge}}(x) \cdot L_{\text{edge}} + C_{\phi,\text{cloud}}(x) \cdot L_{\text{cloud}}+ (1-C_{\phi,\text{edge}}(x)- C_{\phi,\text{cloud}}(x))\cdot L_{\text{human}}.
    % \end{aligned}
\end{equation}
% \begin{multline}\label{cost}
%     \mathcal{L}(\phi|x) = C_{\phi,\text{edge}}(x) \cdot L_{\text{edge}} + C_{\phi,\text{cloud}}(x) \cdot L_{\text{cloud}} \\
%     + (1-C_{\phi,\text{edge}}(x)- C_{\phi,\text{cloud}}(x))\cdot L_{\text{human}}.
% \end{multline}

We assume the conditions $L_{\text{human}} \gg L_{\text{cloud}} > L_{\text{edge}}\geq 0$, which indicate that the cost increases as the query is processed closer to the human expert. 

Our goal is to minimize the average cost, while ensuring that the average misalignment rate is no larger than an acceptable upper bound $\alpha$. This objective is  formalized by the problem
\begin{subequations}\label{constrain_problem}
\begin{align}
    \underset{\phi}{\text{minimize}}\ &\{R_{\mathcal{L}}(\phi)=\mathbb{E}[\mathcal{L}(\phi|x)]\}\\
    \text{subject to}\ &R_{\mathcal{A}}(\phi)=\mathbb{E}[\mathcal{A}(\phi|x,y)]\leq \alpha, \label{constraint}
\end{align}
\end{subequations}
where the expectation $\mathbb{E}[\cdot]$ is over the distribution $P_{xy}$ of the ground-truth test data $(x,y)$.

Problem (\ref{constrain_problem}) cannot be directly addressed since the data distribution $P_{xy}$ is generally unknown. Instead, we assume to have access to independent and identically distributed (i.i.d.)  \emph{calibration data} $\mathcal{D} = \{(x_n, y_n)\}_{n=1}^N \mathop{\sim}\limits^{\text{i.i.d.}} P_{xy}$. With these data, we aim to find thresholds $\phi^*=[\epsilon^*,\lambda^*]$, to be used in the decision  (\ref{cas}), such that the constraint (\ref{constraint}) on the misalignment rate $R_\mathcal{A}(\phi)$ is satisfied with sufficiently high probability, while making a best effort at minimizing the average system cost $R_\mathcal{L}(\phi)$. Specifically, given a user-defined tolerance level $1-\delta$, we impose the alignment requirement
\begin{equation}\label{FWER_goal}
    \Pr[R_\mathcal{A}(\phi^*)\leq \alpha]\geq 1-\delta,
\end{equation}
where the probability is with respect to the calibration dataset $\mathcal{D}$.

\section{Background: Constrained Empirical Risk Minimization}\label{background}

In this section, we review practical solutions for the constrained risk minimization problem (\ref{constrain_problem}) formulated in Sec. \ref{problem formulation} based on the availability of a dataset $\mathcal{D}$. As it will be noted, none of these solutions can provide the target formal guarantee in (\ref{FWER_goal}). 

Conventional solutions to the problem (\ref{constrain_problem}) replace the true averages in (\ref{constrain_problem}) with empirical estimates obtained using dataset $\mathcal{D}$. The resulting optimization problem can be solved using any off-the-shelf methods, such as gradient-based methods or grid search.  To elaborate, given the dataset $\mathcal{D}$, the empirical estimates of risk functions $R_{\mathcal{L}}(\phi)$ and $R_{\mathcal{A}}(\phi)$ in problem (\ref{constrain_problem}) are obtained as
\begin{align} 
        \hat{R}_{\mathcal{L}}(\phi) &= \frac{1}{N}\sum_{n=1}^N \mathcal{L}(\phi|x_n),\label{empirical_cost}\\
        \hat{R}_{\mathcal{A}}(\phi) &= \frac{1}{N}\sum_{n=1}^N \mathcal{A}(\phi|x_n,y_n),\label{empirical_misalignment}
\end{align}
respectively.

 Using these estimates in 
(\ref{constrain_problem}), one can address the resulting problem using grid search. In this case, one limits the optimization to the discrete parameter space of $M\cdot Q$ candidate solutions for some integers $M$ and $Q$  given by
\begin{equation}\label{grid}
    \Phi = \{\{(\epsilon_{m}, \lambda_{q})\}_{m=1}^M\}_{q=1}^Q,
\end{equation}
where 
\begin{equation}
\begin{aligned}
\epsilon_m &= \frac{m-1}{M-1}, \\
\text{and} \ \ \lambda_q &=\frac{q-1}{Q-1},
\end{aligned}
\end{equation}
for $m=1,\ldots,M$ and $q=1,\ldots,Q$. Grid search then solves problem (\ref{constrain_problem}) by exhaustively evaluating the empirical risks $\hat{R}_{\mathcal{L}}(\phi)$ in (\ref{empirical_cost}) and $\hat{R}_{\mathcal{A}}(\phi)$ in (\ref{empirical_misalignment}) for all possible  thresholds $\phi\in\Phi$, and selecting the threshold $\phi^*$ that satisfies the misalignment constraint $\hat{R}_{\mathcal{A}}(\phi)\leq \alpha$ while minimizing the system cost $\hat{R}_{\mathcal{L}}(\phi)$.

Any optimization method based on replacing the true averages with empirical estimates generally does not offer any formal guarantees on the misalignment constraint (\ref{constraint}). In particular, they do not provide any mechanism to meet the constraint (\ref{FWER_goal}). The proposed approach, introduced in the next section, addresses this limitation.

\section{Reliable Edge-Cloud-Expert Cascading}\label{proposed method}
In this section, we propose a novel methodology that optimizes the thresholds $\phi$ in problem (\ref{constrain_problem}) to provide formal statistical reliability guarantees as in (\ref{FWER_goal}). This is done by leveraging recent advances in hyperparameter selection via    \emph{multiple hypothesis testing} (MHT) \cite{ltt,to_cascade_or_not,ltt_test_time_scaling, amir}. {\color{black}{ While this paper focuses on the telecommunications domain, the proposed methodology framework is domain-agnostic and can be applied to any setting in which a calibration set is available.}} We start by introducing ways to evaluate the epistemic uncertainty and confidence scores used in the cascading rule (\ref{cas}), and then we discuss the optimization of the thresholds $\phi$.

\subsection{Epistemic Uncertainty Scores and Confidence Scores}\label{obtain_u_and_c}
While the proposed methodology applies to any choice of epistemic uncertainty and confidence scores, in our evaluations to be presented in Sec.~\ref{simulation}, we adopt the following two approaches. The first is representative of methods applicable to white-box models, while the second can also be implemented with black-box models. {\color{black}{We refer to a model as \emph{white-box} when full access to its internal parameters and architecture is available, enabling direct inspection of model weights and activations. Conversely, a \emph{black-box} model is one accessible only through its input-output interface, e.g., via an API or prompt-based interaction, with no access to internal representations.}}

\subsubsection{Bayesian learning for white-box models}\label{Bayes_white_box}
When white-box access to the model is available, one can leverage schemes that account for epistemic uncertainty at the level of model weights. A principled approach is to adopt a Bayesian formulation for one or more of the layers of the model \cite{simeone2022machine}. A typical implementation of this approach modifies the last layer,
i.e., the classification head, as a Bayesian logistic regression model whose weights $w$ follow the standard Gaussian prior distribution \cite{cambridge_paper}. 

Using a separate held-out dataset $\mathcal{D}^{\text{val}}$, which is distinct from the calibration dataset $\mathcal{D}$, one can obtain an approximate posterior $p(w|\mathcal{D}^{\text{val}})$ over the weights $w$ through variational inference. Alternatively, one can generate samples from distribution $p(w|\mathcal{D}^{\text{val}})$  via Monte Carlo methods \cite{simeone2022machine}. Denoting as $p_w(y|x)$ the model output distribution for a given input $x$, this produces an \emph{ensemble} of models $p_w(y|x)$ with random variables $w\sim p(w|\mathcal{D}^\text{val})$. 

Accordingly, the predictive distribution of the ensemble is given by the expected value  \\$\mathbb{E}_{w\sim p(w|\mathcal{D}^{\text{val}})}[p_w(y|x)]$  with respect to the weights  $w\sim p(w|\mathcal{D}^\text{val})$. The resulting confidence score is evaluated by maximizing over the label index $y$ as \cite{cambridge_paper}
\begin{align} \label{eq:C_wb}
C_{\text{wb}}(x) &= \max_y \mathbb{E}_{w\sim p(w|\mathcal{D}^{\text{val}})}[p_w(y|x)].
\end{align}

The epistemic uncertainty score is then evaluated as the variance of the outputs of the model ensemble as
\begin{align} \label{eq:U_wb}
U_{\text{wb}}(x) &= \mathbb{E}_{w \sim p(w|\mathcal{D}^\text{val})} \Big[ \big(\max_y p_w(y|x)-C_\text{wb}(x)\big)^2 \Big].
\end{align}
This represents a measure of the disagreement among the members $p_w(y|x)$ of the ensemble. 

% where $p_w(y|x)$ represents the output probability for an input $x$ given fixed weights $w$. %  and the averages are evaluated over weights $w$ sampled from the posterior $p(w|\mathcal{D}^{\text{val}})$.
\subsubsection{Prompt-based inference for black-box models}\label{ensembel_black_box}
When only black-box access to the model is available, one can leverage the self-confidence score \cite{ensemble_survey,test_time_scaling_1}, which is produced by the model's self-assessment in response to the input $x$.
 The self-assessment score can be elicited via a \emph{prompt template} $z$ such as  $z=\texttt{Read the question, provide your answer, and}$ $\texttt{your confidence in this answer.}$ $\texttt{Q: [question]}$ 
\cite{ensemble_survey,LLM_self,llm_as_a_judge}. 
% {Common strategy is to design prompt template $z$ such as $z=\texttt{Answer the following with score between 0 and 1.}$ $\texttt{Q: [question]}$ \cite{ensemble_survey,test_time_scaling_1}. }

In order to evaluate the epistemic uncertainty, one can then evaluate the disagreements among the self-confidence scores $C_\text{sc}(x|z_k)$ produced by using $K$ different prompt templates $z_k$ for $k=1,...,K$. 

Overall, in a manner similar to (\ref{eq:C_wb}) and (\ref{eq:U_wb}), the confidence and epistemic uncertainty scores are defined as
% { Following \cite{deepmind}, we design $K$ distinct prompts for multiple-choice question answering task by randomly permuting the ordering of the candidate choices.} 
% targeted prompts. We then implement a prompt variation strategy  \cite{ensemble} which allows us to obtain $K$ multiple self-confidence scores $\{C_{\text{sc}}^k(x)\}_{k=1}^K$.
%to obtain $K$ multiple confidence estimates for each query, and collect the corresponding confidence scores $\{C_{\text{sc}}^k(x)\}_{k=1}^K$. 
% The epistemic and confidence scores can then be defined as
\begin{align} \label{eq:C_bb}
        C_{\text{bb}}(x) &= \frac{1}{K} \sum_{k=1}^{K} C_{\text{sc}}(x|z_k),\\
        \text{and} \ \ U_{\text{bb}}(x) &= \frac{1}{K-1} \sum_{k=1}^K \left(C_{\text{sc}}(x|z_k) - C_{\text{bb}}(x)\right)^2, \label{eq:U_bb}
\end{align}
respectively. Further details on the implementations of the scores (\ref{eq:C_wb})--(\ref{eq:U_wb}) and (\ref{eq:C_bb})--(\ref{eq:U_bb}),  including the design of prompt are provided in Appendix \ref{appendix:confidence score}.

\subsection{Reliable Threshold Optimization}
\label{Reliable Threshold Optimization}
% and an illustration is provided in Fig. \ref{fig:FST}, with details explained in the rest of this section.

% Our methodology builds specifically on the DAGGER of {\textcolor{blue}{[DAGGER]}}, which was tailored for hyparameter selection {\textcolor{blue}{used}} in {\textcolor{blue}{[Amir]}}. Accordingly, 

Using any pre-defined confidence and epistemic uncertainty scores, the proposed methodology, referred to as \emph{MHT-empirical risk
minimization} (MHT-ERM), operates on a grid $\Phi$ of candidate threshold pairs as in (\ref{grid}) to find a pair of thresholds $\phi^*$ that meet the alignment requirement (\ref{FWER_goal}). {\color{black}{It is emphasized that the guarantee (\ref{FWER_goal}) holds irrespective of the quality of the scores. Poorly calibrated scores may increase system cost by causing unnecessary deferral to the human expert, but cannot cause the alignment guarantee to be violated.}} As detailed in the following, MHT-ERM builds on learn-then-test (LTT) \cite{ltt} by leveraging a multi-start fixed sequence testing-based MHT \cite{goeman2014multiple}. The main steps of the procedure are summarized in Algorithm~\ref{alg:mht_optimization}.

\subsubsection{Constructing testing sequences}The starting point of MHT-ERM  is to associate each  threshold  $\phi_{m,q}=(\epsilon_m,\lambda_q)$ in the grid $\Phi$ with the null hypothesis
\begin{equation}\label{null_hypo}
    \mathcal{H}_{m,q}: R_{\mathcal{A}}(\phi_{m,q})> \alpha
\end{equation}
that the choice $\phi_{m,q}\in \Phi$ does not meet the constraint (\ref{constraint}). Rejecting the null hypothesis $\mathcal{H}_{m,q}$ indicates that parameter $\phi_{m,q}$ is deemed to satisfy the misalignment constraint (\ref{constraint}).

{\color{black}{Our testing strategy leverages the observation that increasing the threshold $\lambda$ for a fixed threshold $\epsilon$ can only improve alignment.}} In fact, by (\ref{cas}), a smaller $\lambda$ entails that more queries remain in the edge-cloud system rather than being deferred to a human expert. {\color{black}{Mathematically, the misalignment loss $R_A(\phi_{m,q})$ is non-increasing with respect to $\lambda_q$ and thus with respect to index $q$ from $1$ to $Q$. In contrast, the dependence on the threshold $\epsilon_m$ is non-trivial. In fact, increasing the threshold $\epsilon_m$ routes more queries from the edge to the cloud. In so doing, it may simultaneously redirect queries that would have been handled by a human expert after failing the edge confidence test to the cloud, potentially increasing misalignment. This yields a non-monotonic behavior of the misalignment loss with respect to index $m$.}}

Using this insight, as illustrated in Fig.~\ref{fig:FST}, to test the null hypotheses (\ref{null_hypo}) for $m=1,\ldots,M$ and $q=1,\ldots,Q$, we construct $M$ parallel fixed testing sequences. Each $m$-th sequence corresponds to a fixed threshold $\epsilon_m$, while the confidence thresholds $\lambda$ vary from the highest value $\lambda_Q$ to the lowest value $\lambda_1$ (going from left to right in the first panel of Fig.~\ref{fig:FST}). 
 Accordingly, along each $m$-th sequence, the hyperparameters in the grid are listed in the order of $(\epsilon_m, \lambda_{Q}), (\epsilon_m, \lambda_{Q-1}), \ldots, (\epsilon_m, \lambda_{1})$ where $\lambda_{Q} > \lambda_{Q-1} > \cdots > \lambda_{1}$. 

As discussed next, the testing order in Fig.~\ref{fig:FST} stipulates that, if a pair $\phi_{m,q}$ is deemed not to satisfy (\ref{constraint}), and thus the null hypothesis $\mathcal{H}_{m,q}$ is accepted, the thresholds $\phi_{m,q'}$ with $q'<q$ are not considered as valid candidates, as they are also deemed not to meet (\ref{constraint}). This is because the misalignment loss is known to increase with $q$. %the parallel testing of the $M$ values in the grid.

\begin{figure*}[htbp]
    \centering
    \includegraphics[width=15cm]{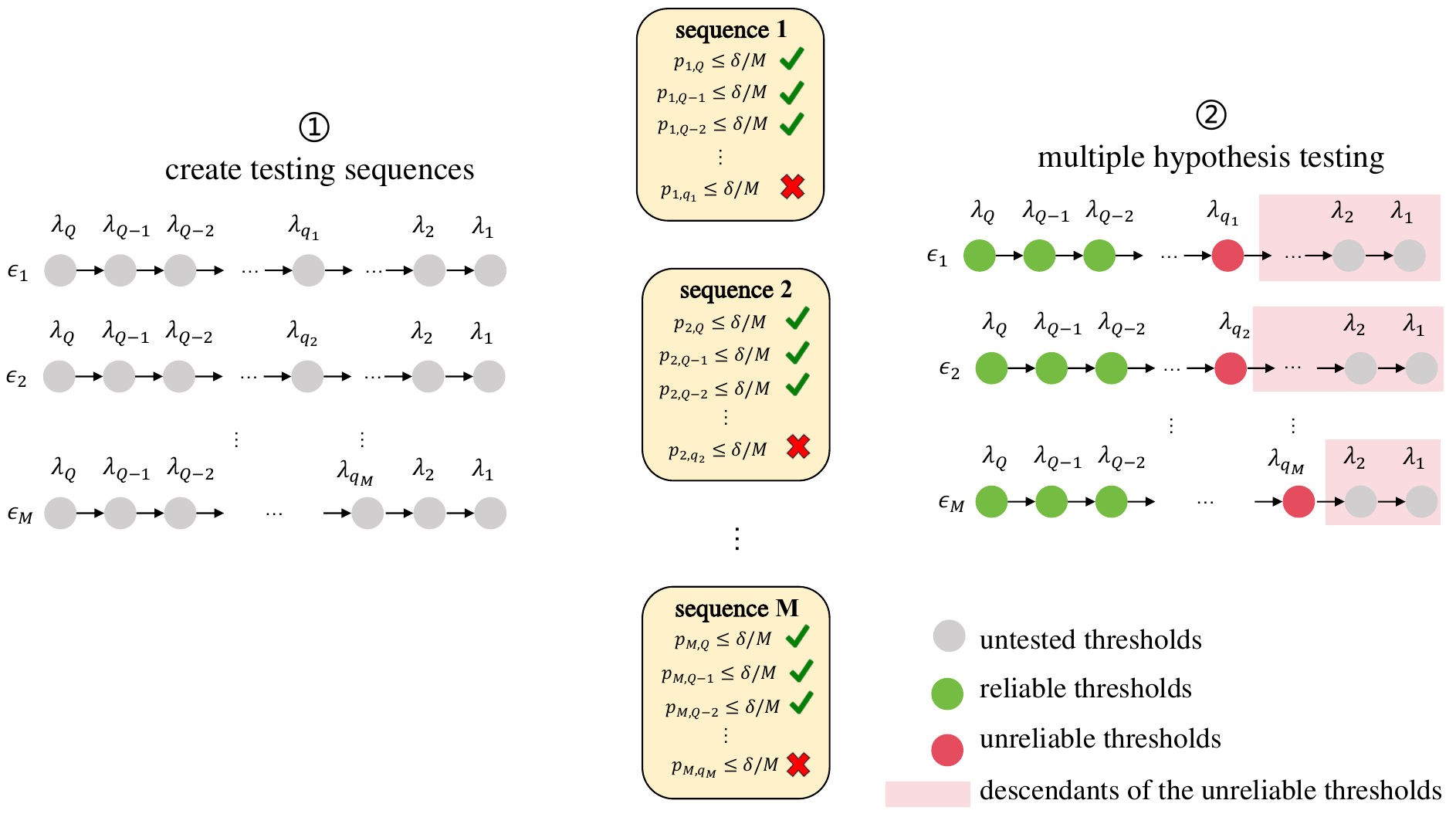}
    \caption{{\color{black}{Illustration of the parallel fixed sequence testing MHT step carried out by the proposed MHT-ERM methodology.}} For each $m$-th sequence corresponding to a value $\epsilon_m$ of the confidence threshold, a pair of thresholds $(\epsilon_m,\lambda_q)$ is tested at each step, starting from $\lambda_Q=1$ and progressively decreasing $q$ through the sequence. The p-value of each pair of thresholds is compared against the risk level $\delta/M$ to assess the reliability of thresholds $(\epsilon_m,\lambda_q)$. All descendants of unreliable thresholds are deemed unreliable.}
    \label{fig:FST}
    % \vspace{2cm}
\end{figure*}

\subsubsection{Multiple hypothesis testing} 
Testing for each hypothesis $\mathcal{H}_{m,q}$ in (\ref{null_hypo}) is carried out by computing a p-value $p_{m,q}$ using the calibration dataset $\mathcal{D}$. By definition, a p-value should be super-uniform under the null $\mathcal{H}_{m,q}$, i.e.,
\begin{equation}\label{eq:superu}
    \Pr[p_{m,q}\leq p|\mathcal{H}_{m,q}]\leq p
\end{equation}
for any probability $p\in[0,1]$. A p-value $p_{m,q}$ is unlikely to be small if the null $\mathcal{H}_{m,q}$ is true, and thus a lower $p$-value $p_{m,q}$ may justify the rejection of the hypothesis $\mathcal{H}_{m,q}$.

Since the alignment loss is bounded in the interval $[0,1]$, a p-value $p_{m,q}$ can be obtained via the Hoeffding inequality\cite{hoeffding1963probability} as
\begin{equation}\label{p_value}
    p_{m,q} = e^{-2N(\alpha-\hat{R}_{\mathcal{A}}(\phi_{m,q}|\mathcal{D}))^2_+},
\end{equation}
where $(\cdot)_+$ denotes the positive part function, i.e., $(x)_+ = \max(0, x)$. Intuitively, the function (\ref{p_value}) decreases as the one-sided margin $(\alpha-\hat{R}_{\mathcal{A}}(\phi_{m,q}|\mathcal{D}))_+$ for the empirical estimate of the constraint (\ref{constraint}) grows larger, indicating the presence of more evidence against the null hypothesis $\mathcal{H}_{m,q}$. 

Using the p-values $\{\{p_{m,q}\}_{m=1}^M\}_{q=1}^Q$, MHT-ERM performs fixed-sequence testing simultaneously on the $M$ parallel chains in Fig.~\ref{fig:FST}.  Specifically, for each $m$-th chain, MHT-ERM tests the threshold $\phi_{m,q}$ starting from the most reliable, with $q=Q$, halting the testing at the largest value of $q$ for which the null $\mathcal{H}_{m,q}$ is not rejected, indicating a choice that does not meet the requirement (\ref{constraint}). 

Mathematically, using Bonferroni correction across the $M$ chains, MHT-ERM tests a parameter $\phi_{m,q}$ via the rule \begin{align}
    p_{m,q} \leq \frac{\delta}{M} \Rightarrow \mathcal{H}_{m,q} \text{ rejected}.
\end{align}
Furthermore, it stops testing at index $q_{m}$ given by
\begin{equation}\label{FST_stopping}
    q_m = \max \bigg\{q=1,\ldots,Q:p_{m,q}> \frac{\delta}{M}\bigg\}.
\end{equation}

Finally, the MHT steps construct a subset $\Phi^* \subseteq \Phi$  as the collection of all threshold pairs $\phi_{m,q}$ that are deemed to be reliable (i.e., for which the null $\mathcal{H}_{m,q}$ is rejected) across all sequences, i.e.,
\begin{equation}
\Phi^* = \bigcup_{m=1}^{M} \{(\epsilon_m,\lambda_{Q}),\ldots,(\epsilon_m,\lambda_{q_{m}+1})\}.
\end{equation}
Note that, whenever $\Phi^*$ is empty, we set $\Phi^* = \{\phi^*=(0,1)\}$, which always ensures the reliability condition $R_{\mathcal{A}}(\phi^*)\leq\alpha$.
% $\Phi^*$ generally is not empty given $\phi=(0,1)$ can always ensure $R_{\mathcal{A}}\leq\alpha$.

\subsubsection{Minimizing the cost} {\color{black}{Finally, MHT-ERM selects the optimal threshold pair  by minimizing the empirical cost risk over the chosen subset $\Phi^*$ according to the problem}}
\begin{equation} \label{eq:phi_star}
    \phi^* = \mathop{\arg\min}\limits_{\phi\in\Phi^*}\ \hat{R}_{\mathcal{L}}(\phi|\mathcal{D}).
\end{equation}

\subsection{Theoretical Properties}
MHT-ERM provides the formal statistical guarantee (\ref{FWER_goal}) that the selected parameter $\phi^*$ satisfies the misalignment constraint (\ref{constraint}) with probability at least $1-\delta$, while simultaneously making a best effort at optimizing computational efficiency.

To elaborate on this, we first show that the subset $\Phi^*\subseteq \Phi$ returned by MHT-ERM is such that every threshold pair in it satisfies the alignment constraint (\ref{constraint}) with high probability as summarized in the following proposition. 

\begin{proposition}[Simultaneous alignment guarantees for subset $\Phi^*$] \label{prop}
With probability at least $1-\delta$, the subset $\Phi^* \subseteq \Phi$ returned by MHT-ERM only contains  threshold pairs that simultaneously satisfy the alignment constraint (\ref{constraint}), i.e., 
\begin{align}\label{proposition 1}
    \Pr\big[ \forall \phi \in \Phi^*: R_\mathcal{A}(\phi) \leq \alpha \big] \geq 1-\delta,
\end{align}
where the probability is over the data set $\mathcal{D}$. 
\end{proposition}
\begin{proof}
See Appendix \ref{proof}.
%Due to (\emph{i}) the super-uniform property of the p-value and (\emph{ii}) the data-independent sequence design in  Fig.~\ref{fig:FST}, the $m$-th subset $\{ (\epsilon_m, \lambda_Q),...,(\epsilon_m, \lambda_{q_m+1}) \}$ cannot contain threshold that admits null hypothesis with probability at least $1-\delta/M$. Applying the union bound over $M$ parallel chains concludes the proof. We refer to \cite{ltt} for further details. 
\end{proof}

Proposition~\ref{prop} implies that any post-selected $\phi \in \Phi^*$ within the subset $\Phi^*$ produced by MHT  satisfies the statistical guarantee (\ref{FWER_goal}), as stated in the following corollary. 

\begin{corollary}[Alignment guarantee for the selected thresholds $\phi^*$] Any choice $\phi \in \Phi^*$, including $\phi^*$ as per (\ref{eq:phi_star}), satisfies the statistical guarantee (\ref{FWER_goal}).
\end{corollary}
\begin{proof}
This follows directly from  Proposition~\ref{prop}.
\end{proof}

\begin{remark}[{\color{black}{Alignment guarantee under distribution shift between calibration dataset and test dataset}}] {\color{black}{The alignment guarantee (26) accounts that calibration and test datasets share the same distribution ${P}_{xy}$. Suppose now that the calibration dataset follows a generally different distribution $\hat{P}_{xy}$, and denote as $\epsilon_{\text{TV}} = D_{\text{TV}}(\hat{P}_{xy}, P_{xy})$ the total variation (TV) distance between the two distributions $\hat{P}_{xy}$ and $P_{xy}$. As shown in Appendix \ref{proof}, the alignment guarantee in (26) remains valid by replacing the target $\alpha$ with the tightened bound $\alpha - \epsilon_{\text{TV}}$ in (\ref{p_value}). Therefore, as long as the TV distance $\epsilon_{\text{TV}}$ is known or can be estimated\cite{zhao2024discriminative,han2018minimax}, the MHT-ERM framework remains valid, albeit at the cost of more conservative threshold selection.}}
\end{remark}

\begin{table}[htp]
	\setlength{\abovecaptionskip}{-2pt}
	\setlength{\belowcaptionskip}{-6pt}
	\begin{algorithm}[H]
		\caption{MHT-ERM}
		\label{alg:mht_optimization}
		{\normalsize
			\begin{algorithmic}[1]
               \State \textbf{Input:} Calibration dataset $\mathcal{D} = \{(x_n, y_n)\}_{n=1}^N$, target upper bound $\alpha$ for the misalignment cost, edge cost $L_\text{edge}$, cloud cost $L_\text{cloud}$, human cost $L_\text{human}$, tolerance level $\delta$, grid sizes $M$ and $Q$ for the grid $\Phi$ in (\ref{grid}) 
               \State \textbf{Output:} Optimized threshold $\phi^*$
               \State \textbf{Initialization:} Reliable subset $\Phi^* = \emptyset$
               \For{$m = 1, 2, \ldots, M$} \Comment{For each epistemic uncertainty threshold in the grid}
                   \State Fix $\epsilon_m = (m-1)/(M-1)$
                   \State Order confidence score thresholds: $\lambda_Q > \lambda_{Q-1} > \cdots > \lambda_1$
                   \For{{\color{black}{$q = Q, Q-1, \ldots, 1$}}} \Comment{Sequential testing along chain $m$}
                       \State Set $\phi_{m,q} = (\epsilon_m, \lambda_q)$
                       \State Compute empirical misalignment risk $\hat{R}_{\mathcal{A}}(\phi_{m,q}|\mathcal{D})$ via  (\ref{empirical_misalignment})
                       \State Compute p-value $p_{m,q}$ via (\ref{p_value})
                       \If{$p_{m,q} \leq \delta/M$} \Comment{Bonferroni correction}
                           \State Add $\phi_{m,q}$ to reliable subset: $\Phi^* = \Phi^* \cup \{\phi_{m,q}\}$
                       \Else
                           \State {break} \Comment{Stop testing in this chain }
                       \EndIf
                   \EndFor
               \EndFor
               \State Compute empirical cost $\hat{R}_{\mathcal{L}}(\phi|\mathcal{D})$ following (\ref{empirical_cost}) for all $\phi \in \Phi^*$
               \State Select optimized threshold: $\phi^* = \arg\min_{\phi\in\Phi^*} \hat{R}_{\mathcal{L}}(\phi|\mathcal{D})$
               \State{\textbf{Return}} $\phi^*$
            \end{algorithmic}}
	\end{algorithm}
\end{table}

\section{Experiments}\label{simulation}
In this section, we evaluate the effectiveness of the reliable edge-cloud-expert cascaded system introduced in Sec. \ref{proposed method} through numerical experiments. The code to reproduce all the results in the following is available at \url{https://github.com/qiushuo0913/reliable_LLM}.

\subsection{Dataset}
For evaluation, we adopt the TeleQnA dataset\footnote{The TeleQnA dataset is publicly available and open-sourced at \url{https://github.com/netop-team/TeleQnA}.}, the first comprehensive benchmark specifically designed to assess telecommunications knowledge \cite{teleQnA}. The dataset comprises $10,000$ multiple-choice questions, each featuring four or five possible answers with only one correct option. The data points are systematically partitioned across five categories: research publications ($45\%$), standards specifications ($20\%$), research overview ($20\%$), standards overview ($10\%$), and lexicon ($5\%$). Accordingly, the corresponding questions span various difficulty levels and cover diverse telecommunications sources, including technical specifications from standardization bodies, such as 3GPP, IEEE, and ITU, research publications from high-impact venues, standards overview documents, and telecommunications lexicon materials.

The TeleQnA dataset directly aligns with the use case of telecommunications query processing discussed in Sec.~\ref{sec:intro} and illustrated in Fig.~\ref{fig:cascaded_model}. In real-world telecommunications systems, technical queries exhibit varying complexity levels corresponding to the dataset's hierarchical categorization. Simple lexicon-based questions (e.g., “What does EIRP stand for?”) represent queries that edge models can handle efficiently with low epistemic uncertainty and high confidence. Moderately complex questions from research overview and standards overview categories, such as questions about software-defined radio techniques or Bluetooth frequency bands, may benefit from cloud model capabilities. The most challenging questions from research publications and standards specifications categories, such as complex MIMO channel calculations or specific 3GPP Release parameters {\color{black}(e.g., ``what is the RRC buffer size for a UE?'')}, often require human expert knowledge due to their technical intricacy. A representative example from the TeleQnA dataset is provided in Appendix \ref{appendix:dataset_example}.

\subsection{Setup}\label{simulation_setting}
We evaluate the proposed MHT-ERM methodology, as well as relevant benchmarks, under two distinct deployment scenarios\footnote{All LLM models used in the experiment are publicly available and open-sourced at \url{https://huggingface.co/}.}:
\begin{itemize}
    \item \emph{Conventional edge-cloud deployment:} The edge employs the Qwen2-1.5B-instruct model, while the cloud utilizes Qwen2-7B-instruct \cite{qwen2_1.5b_hf, qwen2_7b_hf}. Both models operate without reasoning.
    
    \item \emph{Reasoning-enhanced cloud deployment:} The edge maintains the Qwen2-1.5B-instruct model, while the cloud is upgraded to Qwen3-4B with enhanced reasoning capabilities. 
\end{itemize}

For the evaluation of epistemic uncertainty and confidence scores, following \cite{cambridge_paper}, the logistic regressor in the Bayesian ensemble method described in Sec. \ref{obtain_u_and_c} is trained using a held-out data set $\mathcal{D}^\text{val}$ consisting of  $100$ data pairs, while the prompt variation approach uses $K = 10$ different permutations so as to yield multiple self-confidence scores (see Appendix~\ref{appendix:confidence score} for details). {\color{black}{The logistic regressor has a small number of trainable parameters compared to the full LLM, and the Gaussian prior over its weights provides inherent regularization \cite{simeone2022machine,bishop2006pattern}, making the number of data points $|\mathcal{D}^{\text{val}}| = 100$  sufficient to avoid overfitting (see Appendix \ref{additional_experiments} for empirical validation).}} The computational cost incurred by an LLM in the experiments scales linearly with the number of its parameters \cite{cambridge_paper}, and we accordingly set the costs as  $L_{\text{edge}}=1.5$, as well as  $L_{\text{cloud}}=7$ for Qwen2-7B and $L_\text{cloud}=4$ for Qwen3-4B. We also set $L_\text{human}=10$. {\color{black}{We note that this scalar cost model involves no loss of generality, as any multi-dimensional cost objective can be reduced to a scalar via an appropriate scalarization, such as a weighted sum\cite{miettinen1999nonlinear}.}}{ For prompt-based inference using the ensemble method, the cost accounts for multiple model calls, i.e., the costs above are multiplied by $K$ for the edge and cloud models. }

% { I STILL DON'T GET THIS -- DOES IT MEAN THAT THE COST BECOMES $K \cdot L$? We ignore the cost associated with evaluations for confidence and epistemic uncertainty scores when adopting Bayesian learning for white-box models (Sec .~\ref {Bayes_white_box}); while scaling up all the costs when considering prompt-based inference for black-box models (Sec .~\ref {ensembel_black_box}) by multiplying the values by $K$.}

The calibration dataset size is $N = 100$, the grid size is defined by $M=5$ and  $Q=100$, the misalignment upper bound is set to $\alpha = 0.3$, and the tolerance level is $\delta=0.05$. We use test data set $\mathcal{D}_{\text{test}} = \{(x_n, y_n)\}_{n=N+1}^{N+N_{\text{test}}} \mathop{\sim}\limits^{\text{i.i.d.}} P_{xy}$ of size $N_\text{test}=1000$,  and all the results in this section are reported after averaging over $200$ independent experiments. The experiments are carried out with a single A100 GPU.

\subsection{Benchmarks}
To highlight the benefits of MHT-based threshold selection and of the specific testing strategy introduced in Sec. \ref{proposed method}, we consider the following baselines:
\begin{itemize}
    \item \emph{Edge-only}: Edge-only processes all queries exclusively using the edge server.
    \item \emph{Cloud-only}: Cloud-only routes all queries directly to the cloud server.
    \item \emph{Human-only}: Human-only relies exclusively on human operators to answer all queries.
    \item \emph{Conventional-ERM} (C-ERM): As reviewed in Sec.~\ref{background}, C-ERM uses exhaustive grid search over the discrete parameter space $\Phi$ in order to address problem (\ref{constrain_problem}) with empirical estimates (\ref{empirical_misalignment}) used in lieu of the true expectations. 
    \item \emph{MHT-ERM with a Global Bonferroni Correction} (MHT-ERM-B): MHT-ERM-B does not leverage the monotonicity property of the risk $R_\mathcal{A}(\phi)$ with respect to the threshold $\lambda$, hence testing $M Q$ null hypotheses $\{\{\mathcal{H}_{m,q}\}_{m=1}^M\}_{q=1}^Q$ in parallel. In order to ensure the condition (\ref{FWER_goal}), MHT-ERM-B applies Bonferroni correction across all $MQ$ hypotheses, so that each test is carried out with adjusted tolerance level $\delta/(M Q)$. Accordingly, the resulting subset can be written as 
    \begin{equation}
    \Phi^*_\text{B} = \Big\{(\epsilon_m, \lambda_q)\in \Phi : p_{m,q} \leq \frac{\delta}{M Q} \Big\}.
    \end{equation}
    MHT-ERM-B then addresses problem (\ref{eq:phi_star}) with subset $\Phi_\text{B}^*$ in place of $\Phi^*$.
    % employs Bonferroni correction for MHT instead of the proposed fixed sequence testing procedure (see Fig. \ref{fig:FST}). {{Accordingly, MHT-ERM-B constructs a set of validation threshold values by applying uniform risk allocation across all thresholds in the grid $\Phi$. Then, each threshold is tested with risk level $\delta/(M \cdot Q)$ in order to address problem (\ref{constrain_problem})}}.
    
    % \item {Conventional-ERM (Gradient Descent)} (C-ERM (GD)): C-ERM (GD) uses gradient descent optimization as presented in Algorithm~\ref{alg:cascaded_optimization}.
    
\end{itemize}

{\color{black}{The comparison of computational complexity of the aforementioned methods is as follows. C-ERM performs an exhaustive grid search over all $MQ$ candidate thresholds with complexity of order $O(MQ)$. MHT-ERM-B applies a global Bonferroni correction and tests all $MQ$ hypotheses in parallel, also with complexity of order $O(MQ)$. The proposed MHT-ERM exploits the monotonicity of $R_A(\phi_{m,q})$ with respect to $\lambda_q$ to perform fixed-sequence testing along $M$ independent chains with early stopping. The worst-case order of the procedure is still $O(MQ)$, though the average complexity is generally smaller.}}

\subsection{Simulation Results for Conventional Edge-Cloud Deployment}
In this subsection, we evaluate the performance of the proposed MHT-ERM scheme and of the benchmarks discussed above in terms of misalignment rate and system cost by focusing on a conventional edge-cloud deployment with no reasoning. 

\subsubsection{Impact of Score Evaluation}
Fig.~\ref{fig:boxplot} presents misalignment and cost measures for all baselines. The bars show the maximum values after ignoring outliers, which are defined as the values that exceed $1.5$ interquartile range \cite{box_plot}. Standard box plots \cite{box_plot} are also shown in the figure for further detailed statistical information. We have also plotted the target  $1-\delta=0.95$ quantile value in the alignment requirement (\ref{FWER_goal}) as a horizontal line for the misalignment cost. 

%together with box plots across $200$.  The box plots report $1-\delta=0.95$-quantile values using colored horizontal lines in the misalignment rate. This allows one to verify the condition in (\ref{FWER}) by comparing the horizontal quantile markers with the threshold $\alpha=0.3$ (red dashed line).}

% { The figure plots bars representing average values together with box plots across $200$ independent experiments. The box plots report $1-\delta=0.95$-quantile values using whiskers. This allows one to verify the condition in (\ref{FWER}) by comparing the termination point of the lower whiskers with the threshold $\alpha=0.3$ (red dashed line).}

Considering first the extreme baseline strategies that use only one type of informational resource, we observe that edge-only decisions achieve minimal computational cost, but fail to meet the alignment requirement (\ref{FWER_goal}). In contrast, human-only decisions ensure perfect reliability,  but incur maximum cost. Cloud-only decisions represent an intermediate solution with moderate alignment and cost, yet they still fail to meet the requirement (\ref{FWER_goal}).

\begin{figure}[t]
    \centering
    \subfigure[Bayesian learning for white-box models]{
    \includegraphics[width=9cm]{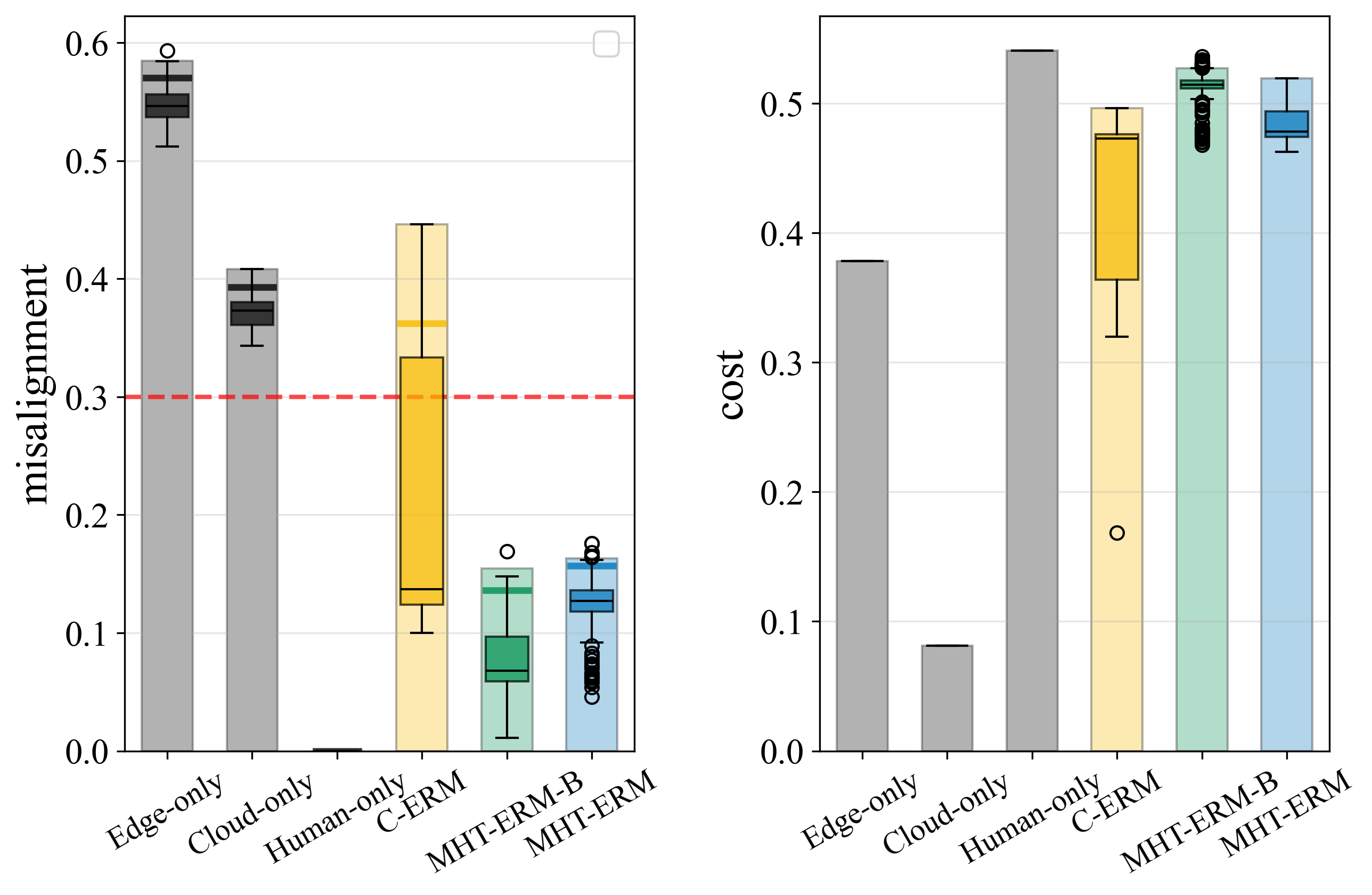}
    \label{fig: bayesian_box}
    \vspace{-0.5cm}
    }
    
    \subfigure[Prompt-based inference for black-box models]
    {
    \centering
    \includegraphics[width=9cm]{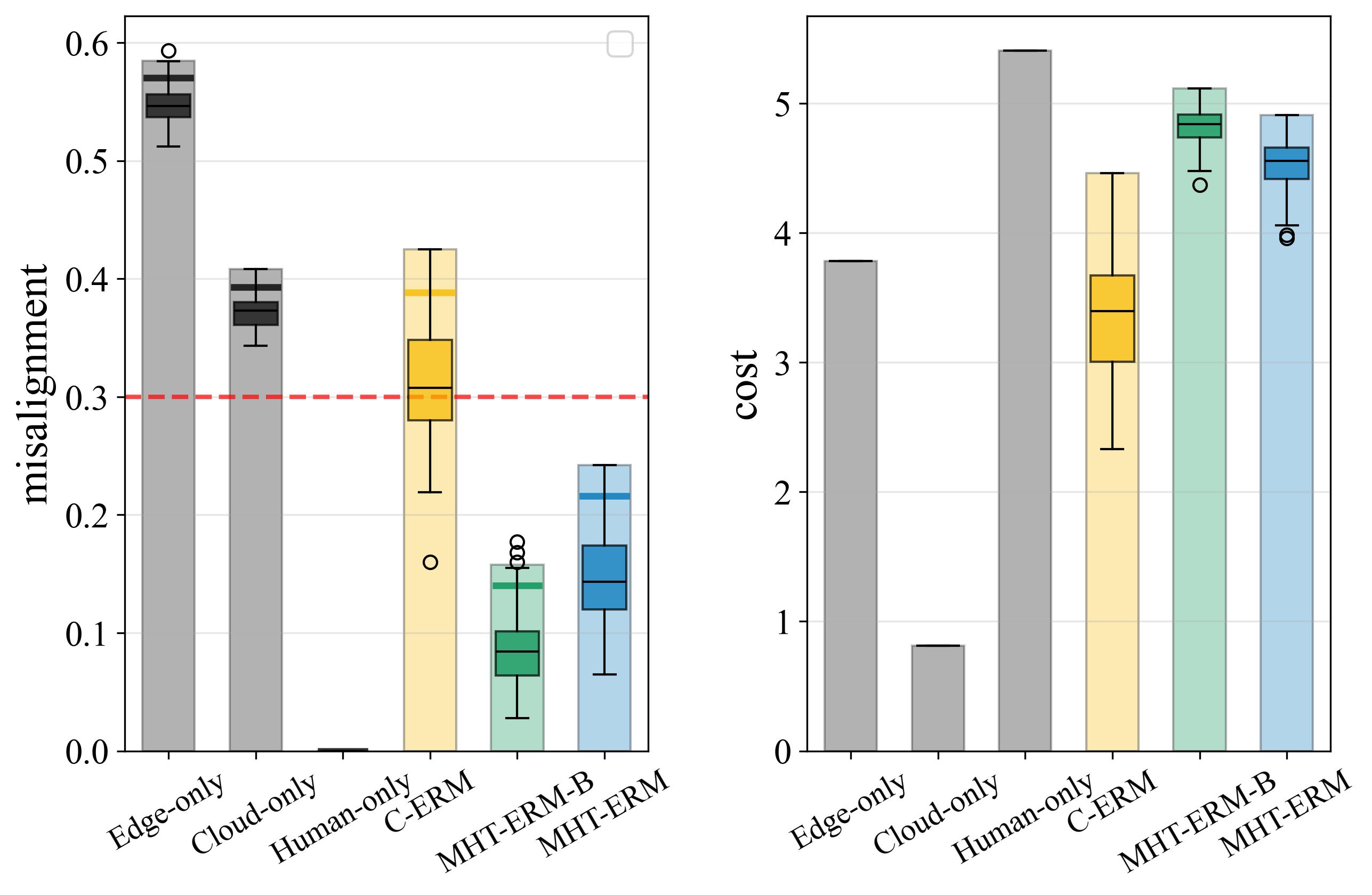}
    \label{fig: ensemble_box}
    }
    \vspace{-0.3cm}
    \caption{{\color{black}{Misalignment and corresponding cost for edge-only, cloud-only, and human-only schemes, as well as for the cascading systems designed via C-ERM, MHT-ERM-B, and MHT-ERM.}} We set the target misalignment risk in (\ref{constraint}) to $\alpha=0.3$ (dashed line) and the target reliability in (\ref{FWER_goal}) to $1-\delta=0.95$. The colored horizontal lines mark the $1-\delta=0.95$-quantile values of the misalignment rate. Maximal values in misalignment performance and cost performance are reported within the $1.5$ interquartile range (IQR) range \cite{box_plot} across $200$ independent experiments.}
    \label{fig:boxplot}
\end{figure}

%{ Moving on to schemes that { implement decision making at edge, cloud, or human level,} both MHT-ERM and MHT-ERM-B seem to consistently ensure misalignment requirements, providing reliable statistical guarantees, while C-ERM generally cannot. In fact, C-ERM does not come with any mechanism to provide theoretical guarantees for constraint violation, while its optimization on the calibration set and evaluation on the test set leads to inadequate misalignment with the desired risk control objectives. Furthermore, MHT-ERM-B is more conservative than MHT-ERM, yielding smaller candidate sets, thus { having} higher costs.} 

Cascading methods, designed using any of the schemes C-ERM, MHT-ERM-B, and MHT-ERM, achieve a lower cost than human-only decisions. However, only the proposed MHT-based schemes satisfy the alignment constraint (\ref{FWER_goal}). Furthermore, the proposed MHT-ERM is seen to achieve a lower cost than MHT-ERM-B, demonstrating the importance of incorporating domain knowledge in the design of the testing strategy, as done by the proposed approach illustrated in Fig.~\ref{fig:FST}.

Comparing Fig.~\ref{fig:boxplot}(a) and (b) allows us to analyze the impact of the choice of specific epistemic uncertainty and confidence scores.  Specifically, scores obtained via Bayesian learning for white-box models are seen to achieve lower misalignment rates compared to prompt-based inference for black-box models. This superiority stems from direct access to model logits, which enable a more precise uncertainty estimation through the variance of the posterior distribution over model weights. The resulting uncertainty measures better reflect the model's epistemic knowledge gaps, leading to more accurate query processing decisions and improved alignment with expert judgments.

%In contrast, the black-box approach relies on self-evaluation prompts to estimate confidence, which provides a simpler implementation requiring only API access. However, to obtain meaningful uncertainty estimates, this approach necessitates ensemble averaging across multiple prompt evaluations, resulting in $K$-fold computational overhead.}

% a comprehensive performance comparison across all baselines, where the $1-\delta=0.95$-quantile values in misalignment performance are displayed as bars to verify the guarantee of condition in (\ref{FWER}), and maximum values in cost performance are shown within the $1.5$ interquartile range (IQR) \cite{box_plot} across $200$ independent experiments. The target coverage misalignment rate is $\alpha=0.3$ (red dashed line).

\subsubsection{Impact of Calibration Data Size}
All cascading methods rely on the availability of the held-out calibration data $\mathcal{D}$ to optimize the threshold $\phi$. Here we investigate the impact of the size of the calibration data, $N$.  As shown in Fig. \ref{fig:impact_cal_size}, both MHT-based methods and C-ERM demonstrate improved performance with larger calibration sets in terms of both misalignment and cost. Specifically, C-ERM starts to satisfy the alignment constraint with a sufficiently large amount of calibration data ($N > 200$), but it exceeds the target misalignment  (attaining a misalignment value $0.4 > \alpha=0.3$) when fewer calibration data points are available ($N=10$). In contrast, MHT-based schemes always satisfy the alignment constraint irrespective of the size of the calibration data. 

\subsubsection{Impact of Misalignment Target $\alpha$}
{\color{black}{We now turn our attention to the impact of the misalignment target level $\alpha$. We set the size of calibration data to $N=100$. Fig. \ref{fig:impact_constraint} shows the results for $\alpha$ ranging from $0.1$ to $0.7$. Cascading via MHT-based threshold selection always satisfies the alignment constraint irrespective of the target level $\alpha$. This is unlike C-ERM, which fails to satisfy the constraint in the regime of practical interest ($\alpha < 0.5$). We note that for $\alpha$ below the error rate of the cloud model (approximately $0.3$ for Qwen2-7B-instruct), both MHT-based methods defer most queries to the human expert, and the sequential testing advantage of MHT-ERM over MHT-ERM-B becomes more pronounced as $\alpha$ increases into the non-trivially binding regime.}}

\begin{figure}[htp]
    \centering
    \subfigure[Bayesian learning for white-box models]{
    \includegraphics[width=7.5cm]{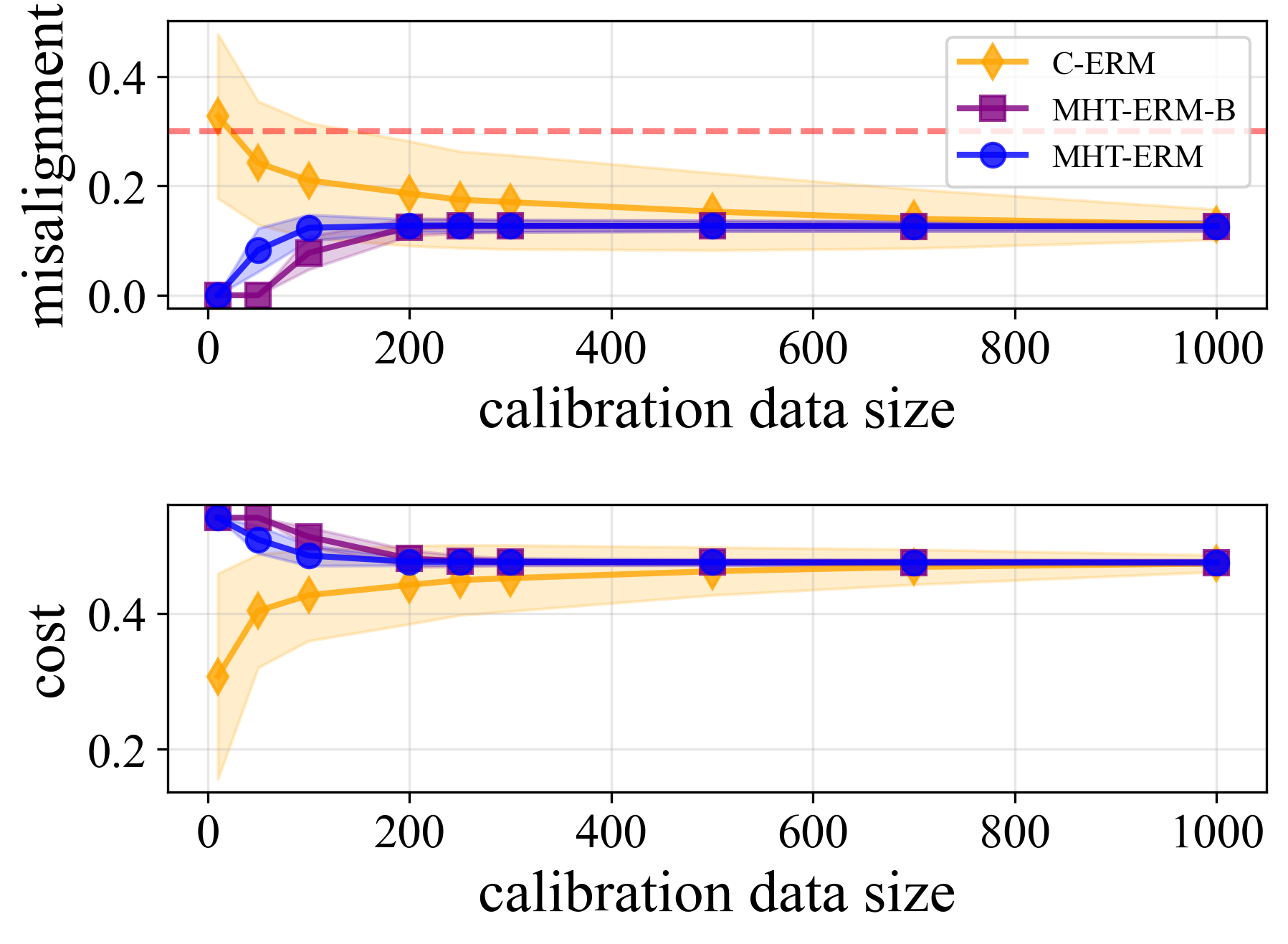}
    \vspace{-0.5cm}
    }
    % \label{fig: bayesian_}
    \subfigure[Prompt-based inference for black-box models]
    {
    \centering
    \includegraphics[width=7.5cm]{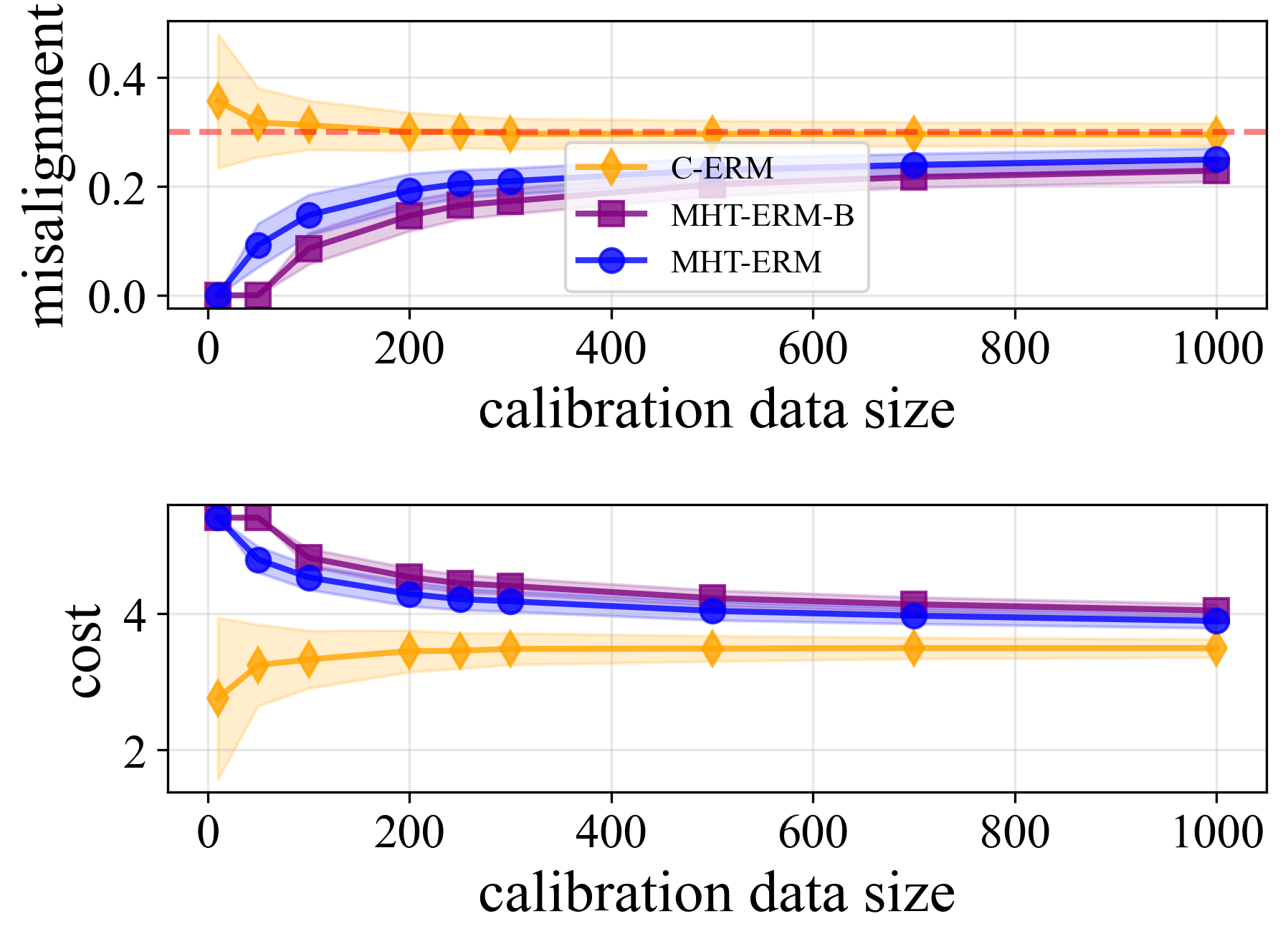}
    % \label{fig: ensemble_box}
    }
    % \vspace{-0.3cm}
    \caption{Misalignment and corresponding cost for the cascading systems with thresholds chosen via C-ERM, MHT-ERM-B, and  MHT-ERM under different values of calibration dataset size. We set the target misalignment risk in (\ref{constraint}) to $\alpha=0.3$ (dashed line) and target reliability in (\ref{FWER_goal}) to $1-\delta=0.95$. The results are averaged over $200$ independent experiments (shaded bar on plots shows one standard deviation on both sides).}
    \label{fig:impact_cal_size}
\end{figure}

\begin{figure}[htp]
    \centering
    \subfigure[Bayesian learning for white-box models]{
    \includegraphics[width=7.5cm]{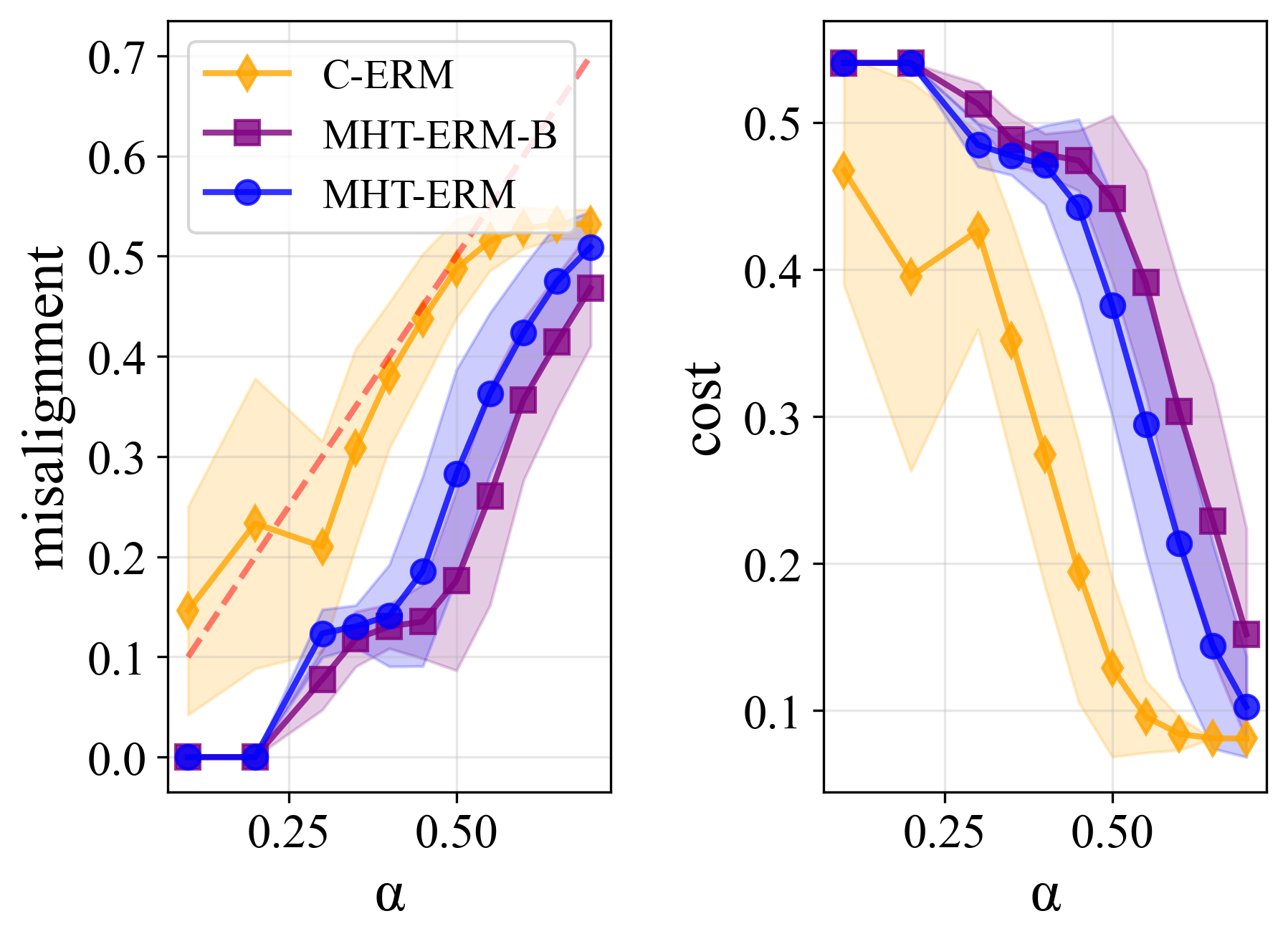}
    \vspace{-0.5cm}
    }
    % \label{fig: bayesian_}
    \subfigure[Prompt-based inference for black-box models]
    {
    \centering
    \includegraphics[width=7.5cm]{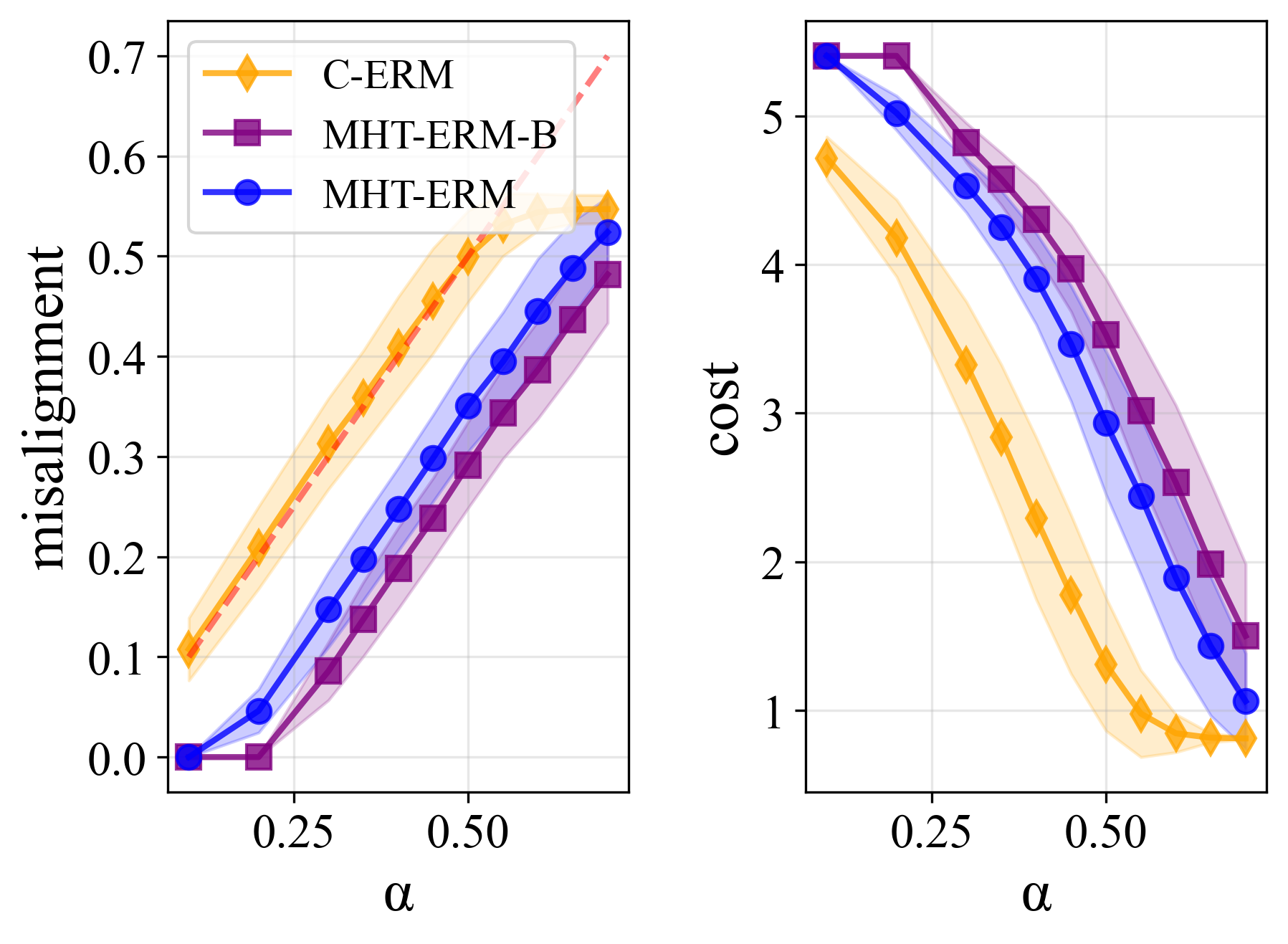}
    % \label{fig: ensemble_box}
    }
    % \vspace{-0.3cm}
    \caption{{\color{black}{Misalignment and corresponding cost for the cascading systems with thresholds chosen via C-ERM, MHT-ERM-B, and MHT-ERM under different values of misalignment upper bound.}} We set the target misalignment risk in (\ref{constraint}) from $\alpha=0.1 $ to $\alpha=0.7$ (dashed line) and target reliability in (\ref{FWER_goal}) to $1-\delta=0.95$. The results are averaged over $200$ independent experiments (shaded bar on plots shows one standard deviation on both sides).}
    \label{fig:impact_constraint}
\end{figure}

% As shown in Fig. \ref{fig:impact_constraint}, only the MHT-based methods (MHT-ERM and MHT-ERM-B) seem to consistently satisfy the alignment requirements across the entire range of $\alpha$ values, as evidenced by both the blue and purple curves (including their shaded confidence regions) remaining entirely below the red dashed line. In contrast, the C-ERM exhibits constraint violations when the misalignment bounds are restrictive, given that its shaded confidence region exceeds the red dashed line. As the misalignment requirements are progressively relaxed (increasing $\alpha$), the C-ERM eventually satisfies the misalignment constraints.

\begin{figure}[htp]
    \centering
    \subfigure[Bayesian learning for white-box models]{
    \includegraphics[width=7.5cm]{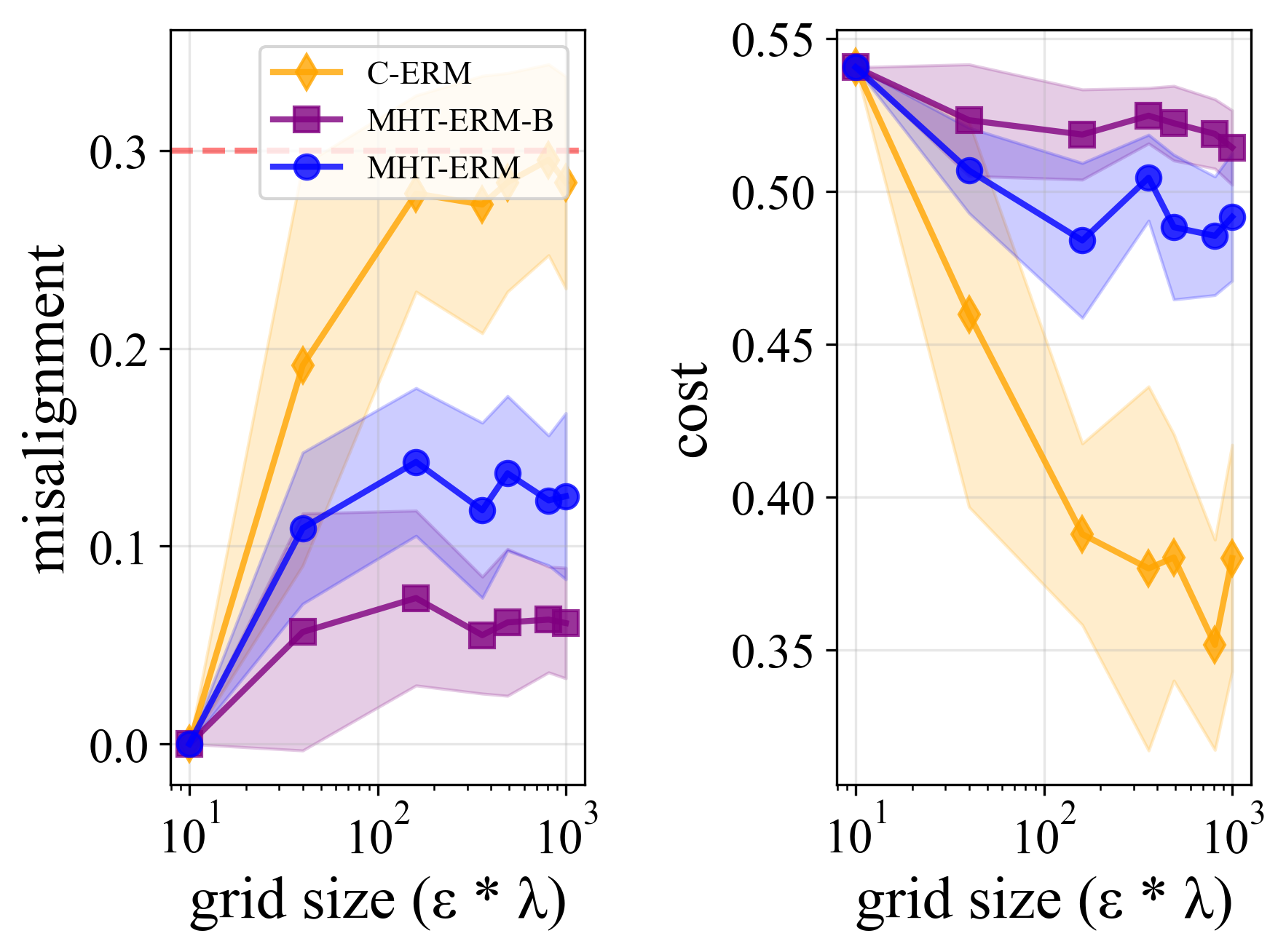}
    \vspace{-0.5cm}
    }
    % \label{fig: bayesian_}
    \subfigure[Prompt-based inference for black-box models]
    {
    \centering
    \includegraphics[width=7.5cm]{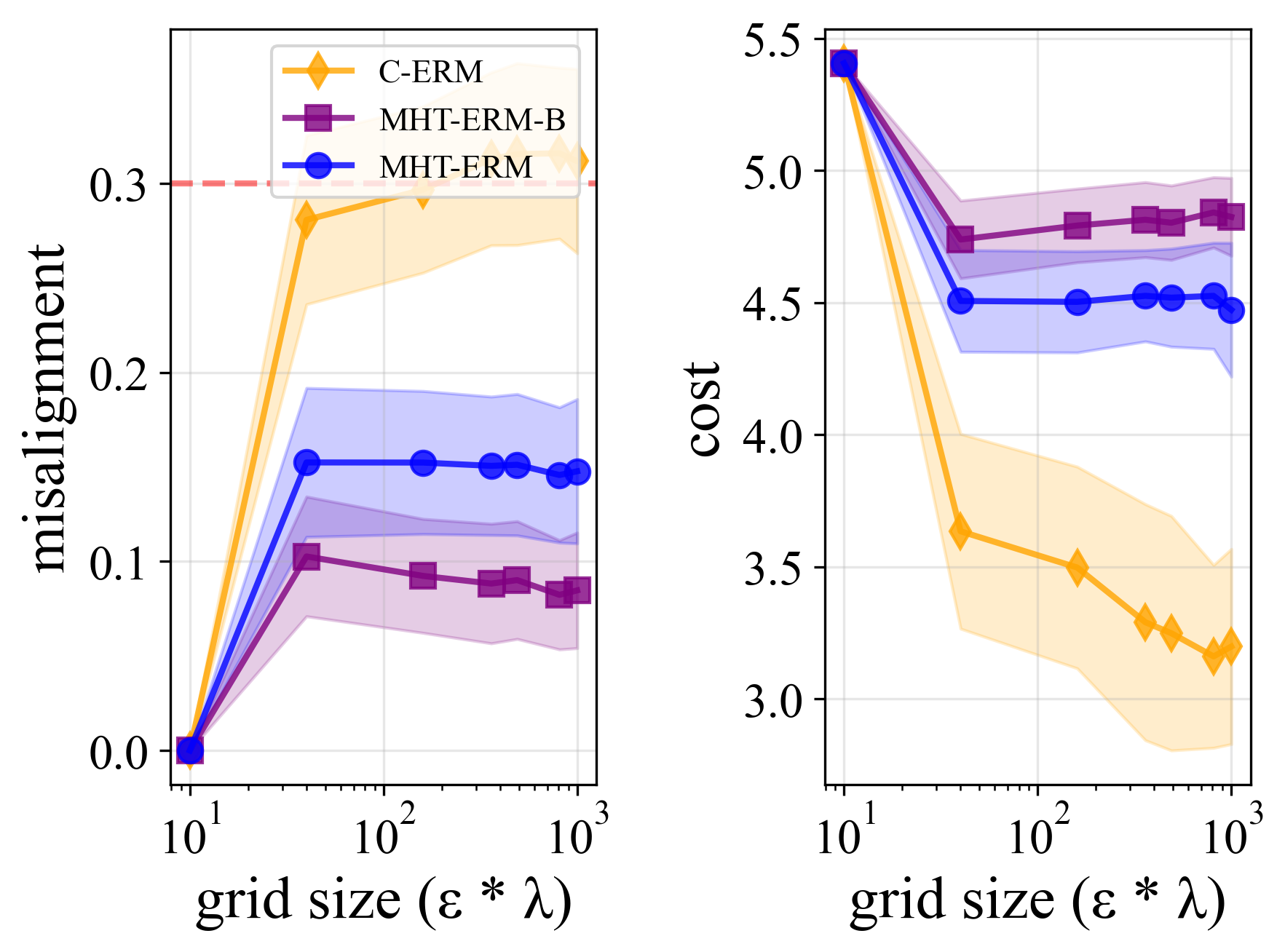}
    % \label{fig: ensemble_box}
    }
    % \vspace{-0.3cm}
    \caption{Misalignment and corresponding cost for the cascading systems with thresholds chosen via C-ERM, MHT-ERM-B, and MHT-ERM under different values of grid sizes. We set the target misalignment risk in (\ref{constraint}) to $\alpha=0.3$ (dashed line) and target reliability in (\ref{FWER_goal}) to $1-\delta=0.95$. The results are averaged over $200$ independent experiments (shaded bar on plots shows one standard deviation on both sides).}
    \label{fig:impact_grid_size}
\end{figure}

\subsubsection{Impact of Grid Size}
Finally, we investigate the impact of the grid size for ERM-based approaches in Fig. \ref{fig:impact_grid_size}. We set $\alpha=0.3$ and $N=100$. It is observed that the cost gap between MHT-ERM and MHT-ERM-B becomes more pronounced with increased grid size, demonstrating the importance of the proposed sequential design. Furthermore, as the grid size expands, C-ERM suffers from overfitting to the calibration data,  leading to worse constraint violation performance.

{\color{black}{To further validate the generalizability of the proposed framework, additional experiments on ORAN-Bench-13K \cite{gajjar2025oranbench, saenko2026twinpass} are presented in Appendix \ref{additional_experiments}.}}

% naturally increases with larger grid sizes, reflecting the increasingly conservative nature of Bonferroni correction when applied to a growing number of hypotheses. 

\subsection{Simulation Results for Reasoning-Enhanced Cloud Deployment}\label{sim:reasoning}

We now adopt a reasoning model at the cloud by controlling the thinking budget of the Qwen3-4B reasoning model, which is measured by the number of thinking tokens in the reasoning procedure.  Fig. \ref{fig:qwen3} presents the performance of MHT-ERM as a function of the thinking budget when using Bayesian learning for evaluating the epistemic uncertainty and confidence scores. The left panel of the figure also shows the accuracy improvement of the cloud model with increased thinking budget, which increases from 70.4\% to 71.6\%.

The left panel shows the misalignment rate as a function of the thinking budget. The blue bars report the mean values of the misalignment rate,   with error bars indicating the corresponding standard deviation. The purple horizontal lines mark the $1-\delta = 0.95$ quantile values of the misalignment rate. It is verified that, irrespective of the thinking budget, MHT-ERM achieves lower misalignment than the target $\alpha=0.25$ (red dashed line) demonstrating its statistical validity. Furthermore, MHT-ERM exhibits enhanced alignment under increased thinking budget with a marginal cost increase. 

% {\color{red} here $L_\text{cloud}$ may need to be  changed as per thinking budget, but let's not change anything and keep it as vague as possible :) }
%This allows one to verify the condition in (\ref{FWER}) by comparing the purple horizontal lines with the threshold $\alpha=0.25$ (red dashed line). As the thinking budget increases, misalignment rates decrease while system costs remain relatively stable. This can be explained by the underlying mechanism shown in the green curve: the cloud model accuracy exhibits a clear upward trend with increased budget.

\begin{figure}[htbp]
    \centering
    \includegraphics[width=9cm]{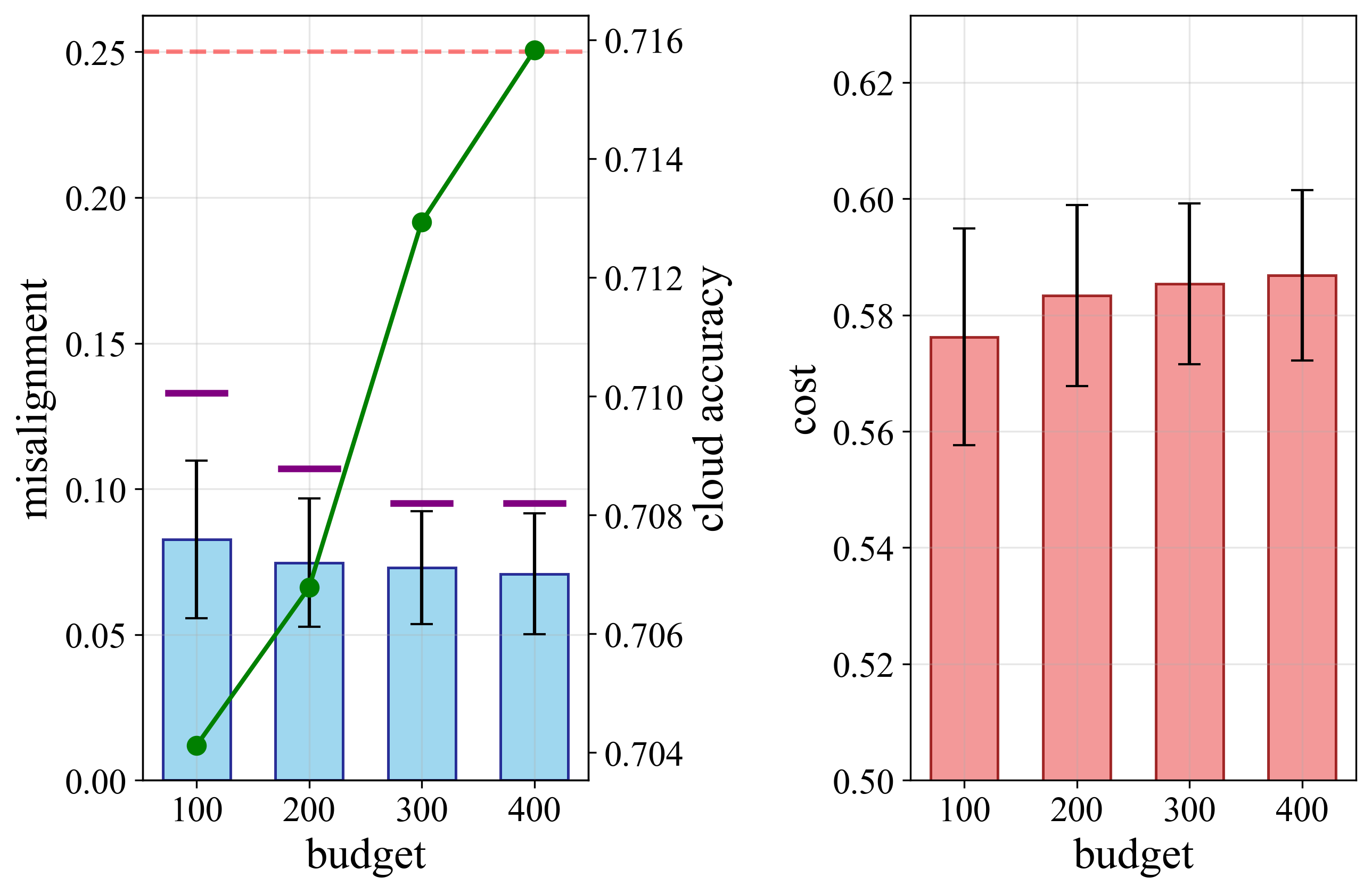}
    \caption{Misalignment and corresponding cost for the cascading systems with thresholds chosen via MHT-ERM under a reasoning-enhanced cloud deployment as a function of thinking budget for $\alpha=0.25$ (dashed line in the left panel). The bars, reporting the mean, are augmented with $95\%$-quantile (purple lines). The results are averaged over $200 $ independent experiments, with error bar indicating one standard deviation.}
    \label{fig:qwen3}
    % \vspace{2cm}
\end{figure}

\section{Conclusions}\label{conclusion}
Large language models are emerging as key enablers of automation in telecommunications, yet their deployment must carefully balance inference cost, latency, and reliability constraints. This work develops a statistically principled cascading framework in which lightweight edge models, powerful cloud models, and human experts collaborate through rigorous routing decisions to minimize average processing costs while guaranteeing alignment with expert judgments. The core contribution lies in reformulating threshold selection as a multiple hypothesis testing (MHT)  problem, thereby providing finite-sample guarantees on misalignment risk.  Experimental validation on the TeleQnA dataset demonstrates the effectiveness of our approach, showing cost reductions while maintaining target reliability levels.

Future research directions include investigating adaptive threshold mechanisms that can respond to evolving data distributions inherent in dynamic wireless environments, {\color{black}{designing query-dependent thresholds that can adapt the routing decision to the specific type or difficulty of each query, considering tier-specific thresholds that account for the heterogeneous capabilities of the edge and cloud models,}} exploring multi-objective optimization scenarios, and addressing challenges associated with limited calibration datasets.

% { Please check the format of all references that follows to satisfy the usual IEEE requirements}

\bibliographystyle{ieeetr}
\bibliography{reference}
% % \clearpage
% \begin{thebibliography}{1}
% \small

% \bibitem{}
% \bibitem{scheduling1}
% R. M. Sohaib, S. T. Shah, O. Onireti, Y. Sambo, Q. H. Abbasi, and M. A. Imran, “DRL-based joint resource scheduling of eMBB and URLLC in O-RAN,” in \emph{Proc. IEEE Int. Conf. on Commun. (ICC)}, Denver, 2024.

% \bibitem{nokia}
% https://github.com/nokia/wireless-suite/tree/master/wireless

% \end{thebibliography}

% \newpage
\appendices

\section{Implementation Details on Confidence Score}\label{appendix:confidence score}
{\color{black}{This appendix provides detailed implementation procedures for both the white-box Bayesian learning method and the black-box self-confidence score estimation method used in our experiments.}}

\begin{tcolorbox}[colback=gray!10!white,colframe=gray!75!black,title={{Box A-1: Prompt Templates of Self-Evaluation Implementation for Qwen2-1.5B}}, label={box:qwen2_prompt}]
{\textbf{Stage 1: Answer Generation}}\\
Please provide the answers to the following telecommunications related multiple choice questions. The questions will be in a JSON format, the answers must also be in a JSON format as follows:\\
\{\\
“question 1”: \{\\
“question”: question,\\
“answer”: “option \{answer id\}: \{answer string\}”\\
\},\\
...\\
\}\\
Here are the questions: [question]

{\textbf{Stage 2: Confidence Evaluation}}\\
Please act as an impartial telecommunications expert and evaluate the quality of the answer provided by an AI assistant to the user question displayed below. Your evaluation should assess the probability that the given answer to a telecommunications question is correct. Return ONLY a number BETWEEN 0 AND 1, where:
\begin{itemize}
\item 0 means definitely incorrect
\item 1 means definitely correct
\end{itemize}
Question: [question]\\
Answer: [answer]\\
Return your response in the following JSON format ONLY:\\ 
\{“probability”: 0.X\}
\end{tcolorbox}

{\color{black}{For the Bayesian learning method, in our multiple-choice setting, the model output $y$ corresponds to the index of the selected option (e.g., ``1'', ``2'', ``3'', ``4 '', or ``5''). Accordingly, the quantity $P_w(y|x)$ is evaluated as the conditional probability of the \emph{single output token} corresponding to the option index given the past tokens, rather than the joint probability of a full answer sentence. Therefore, this approach avoids any length-related bias in the confidence score estimation.}}

\begin{tcolorbox}[colback=gray!10!white,colframe=gray!75!black,title={{Box A-2: Prompt Templates of Self-Evaluation Implementation for Qwen3-4B}}, label={box:qwen3_prompt}]
{\textbf{Stage 1: Answer Generation}}\\
Please provide the answers to the following telecommunications related multiple choice questions. The questions will be in a JSON format, the answers must also be in a JSON format as follows:\\
\{\\
“question 1”: \{\\
“question”: question,\\
“answer”: “option \{answer id\}: \{answer string\}”\\
\},\\
...\\
\}\\
Here are the questions: [question]\\
Please think step by step before answering.
{\textbf{Stage 2: Confidence Evaluation}}\\
Please think step by step to thoroughly assess your previous answer. You need to evaluate how likely your answer is correct from reasoning quality, potential uncertainties, and alternative possibilities.\\
Confidence Analysis:

...

[confidence analysis]

...

Based on my thinking above, I will now output my confidence score between 0 and 1, where 0 means completely uncertain and 1 means completely certain. \\
Format: \{“confidence score”: [number]\}
\end{tcolorbox}

% \subsection{Log-Loss Confidence Score}
% To obtain the prediction $p(y|x)$, we extract the probability distribution over the vocabulary at the specific generation step where the model produces its  selection. For multiple-choice questions where the model generates responses in the format “option [number]”, the implementation applies the following procedure:
% \begin{itemize}
%     \item \emph{Token position identification}: When the model outputs text like “option 1” or “option 2”, we identify the character position where the option number appears in the generated text (e.g., the position of “1” in “option 1”).

%     \item \emph{Token index mapping}: We map this character position back to the corresponding token index in the generation sequence to locate the exact generation step where the option number token is produced.

%     \item \emph{Logits extraction}: At this specific generation step, we extract the logits $l_i$ for all possible option number tokens from the model's output.

%     \item \emph{Probability computation}: The softmax probabilities are computed as
%     \begin{equation}
%         p(y|x) = \frac{e^{l_i/T}}{\sum_{j \in \{1,\ldots,|Y|\}} e^{l_j/T}}
%     \end{equation}
% where $T$ is the temperature.
% \end{itemize} 

% By doing this, it provides a direct measure of prediction uncertainty based on the model's internal token-level representations for the option number tokens.

% { I did not understand this subsection -- isn't this just extracting the logits of the classification head? I don't see the need for the description above}

% \subsection{Self-Confidence Score}
{\color{black}{For the black-box self-confidence score estimation method, in a manner similar to \cite{llm_as_a_judge}, we implement a two-stage process for the Qwen2-1.5B-instruct model, where the model first generates its answer, and then evaluates its confidence in that answer through a separate prompt. This is done by using the following templates in Box A-1.}}

For the Qwen3-4B model, inspired by \cite{test_time_scaling_1,test_time_scaling_2}, we design a test-time scaling approach that dynamically allocates computational budget between answer generation and confidence analysis using the following templates in Box A-2.

%\newpage
\begin{tcolorbox}[colback=orange!10!white,colframe=orange!75!black,title={{Box A-3: Test-Time Scaling Implementation for Qwen3-4B}}]
{\textbf{Stage 1: Answer Generation Test-Time Scaling}}\\
\texttt{<think>}\\
$\ldots$ [thinking tokens for answering during reasoning process] $\ldots$
\textcolor{orange}{Wait...} [inserted when the budget for answering is not exhausted]\\
\textcolor{orange}{\texttt{</think>}} [inserted when the budget for answering is exhausted]

\{\\
“question 1”: \{\\
“question”: question,\\
“answer”: “option \{answer id\}: \{answer string\}”\\
\},\\
...\\
\}

{\textbf{Stage 2: Confidence Analysis Test-Time Scaling}}\\
\texttt{<think>}\\
$\ldots$ [thinking tokens for confidence reasoning during reasoning process]$\ldots$
\textcolor{orange}{Wait...} [inserted when the budget for confidence reasoning is not exhausted]\
\
\emph{Based on my thinking above, I will output my confidence score between 0 and 1.\
Format: \{“confidence score”: [number]\}}\\
\textcolor{orange}{\texttt{</think>}} [inserted when the budget for confidence reasoning is exhausted]

confidence score: 0.X
\end{tcolorbox}

% \vspace{1cm}
The Box A-3 illustrates the mechanism for controlling the thinking budget in our test-time scaling implementation of Section \ref{sim:reasoning}. The total thinking budget is evenly divided between two stages: answer generation and confidence analysis. In both stages, we dynamically control the number of reasoning tokens by inserting \texttt{Wait...} tokens when the allocated budget for that stage is not exhausted, allowing the model to continue reasoning. When the budget is depleted, a \texttt{</think>} tag is inserted to terminate the reasoning process, prompting the model to output the answer in JSON format (Stage 1) or a confidence score between $0$ and $1$ (Stage 2).

% \newpage
\section{ Proof of Proposition 1 and Remark 1}\label{proof}

\subsection{Proof of Proposition 1}

By the Hoeffding-based construction in (\ref{p_value}), each p-value $p_{m,q}$ is super-uniform under its null hypothesis $\mathcal{H}_{m,q} : R_A(\phi_{m,q}) > \alpha$. That is, for any $p \in [0,1]$, we have the inequality (\ref{eq:superu}). In particular, setting $p = \delta/M$, we have
\begin{equation}\label{eq:superu_appendix}
\Pr\left[p_{m,q} \leq \frac{\delta}{M} \;\Big|\; \mathcal{H}_{m,q}\right] \leq \frac{\delta}{M}.
\end{equation}

Consider the $m$-th sequence, it tests hypotheses $\mathcal{H}_{m,Q}, \mathcal{H}_{m,Q-1}, \ldots, \mathcal{H}_{m,1}$ in this order. Let $q^*$ be the index of the first true null hypothesis in the sequence if one exists, i.e., $q_m^* = \max\{q : \mathcal{H}_{m,q} \text{ is true}\}$. If such true null hypothesis does not exist, any testing procedure is trivially reliable thus we will henceforth consider the case where $q_m^*$ exists. Now, let $\mathcal{E}_m$ denote the event that at least one true null hypothesis is incorrectly rejected during the $m$-th sequence testing (\ref{FST_stopping}). By noting that $\mathcal{E}_m$ can only happen whenever the first true null  $\mathcal{H}_{m,q_m^*}$ had been rejected, we have $\Pr[\mathcal{E}_m]= \Pr[p_{m,q_m^*}\leq \delta/M | \mathcal{H}_{m,q_m^*}]\leq \delta/M$. 

% Therefore, the probability of any false rejection in the sequence is bounded by the probability of rejecting the first true null:

% Thanks to the monotonicity property discussed in Section \ref{Reliable Threshold Optimization}, if $\mathcal{H}_{m,q^*}$ is true, i.e.,$R_A(\phi_{m,q^*}) > \alpha$, then all hypotheses $\mathcal{H}_{m,q}$ with $q < q^*$ are also true because $R_A(\phi_{m,q})$ is non-decreasing in $q$.

% A false rejection in the $m$-th sequence occurs if and only if the testing procedure rejects this first true null $\mathcal{H}_{m,q^*}$, which happens when $p_{m,q^*} \leq \delta/M$. Crucially, once we fail to reject $\mathcal{H}_{m,q^*}$ (i.e., when $p_{m,q^*} > \delta/M$), the sequential testing stops, and all subsequent hypotheses $\mathcal{H}_{m,q}$ with $q < q^*$ are not tested and hence not rejected. 

% \begin{align}
% \Pr[\mathcal{E}_m] &= \Pr[\text{reject the first true null in sequence } m] \\
% &= \Pr[p_{m,q^*} \leq \delta/M \mid \mathcal{H}_{m,q^*}] \\
% &\leq \frac{\delta}{M},
% \end{align}
% where the inequality follows from (\ref{eq:superu_appendix}).

Now, let $\mathcal{E}$ denote the event that at least one false rejection occurs across all sequences. By the union bound over the $M$ parallel sequences:
\begin{equation}
\Pr[\mathcal{E}] = \Pr\left[\bigcup_{m=1}^{M} \mathcal{E}_m\right] \leq \sum_{m=1}^{M} \Pr[\mathcal{E}_m] \leq M \cdot \frac{\delta}{M} = \delta.
\end{equation}

By noting that (\ref{proposition 1}) in Proposition~\ref{prop} describes the complement event $\mathcal{E}^c$ where the rejected hypothses $\phi^*$ do not contain any true null hypotheses, we have 
% To establish (\ref{proposition 1}), we observe that the complement event $\mathcal{E}^c$ corresponds to the situation where no false rejection occurs. Since $\Phi^*$ contains exactly those threshold pairs $\phi_{m,q}$ for which the null hypothesis $\mathcal{H}_{m,q}$ was rejected, we have
% \begin{equation}
% \mathcal{E}^c = \left\{\forall \phi \in \Phi^* : R_A(\phi) \leq \alpha\right\}.
% \end{equation}
% Therefore,
\begin{equation}
\Pr\left[\forall \phi \in \Phi^* : R_A(\phi) \leq \alpha\right] = \Pr[\mathcal{E}^c] = 1 - \Pr[\mathcal{E}] \geq 1 - \delta,
\end{equation}
which completes the proof.
\qed

\subsection{{\color{black}{Proof of Remark 1}}}

{\color{black}{For any bounded misalignment loss $\mathcal{A} \in \{0,1\}$, the TV distance bound \cite[Eq.~(4.1)]{levin2017markov} gives
\begin{equation}
\begin{aligned}
    R_A(\phi) &= \mathbb{E}_{P_{xy}}[A(\phi|x,y)],\\ 
    &\leq \mathbb{E}_{\hat{P}_{xy}}[A(\phi|x,y)] + \epsilon_{\text{TV}} = \hat{R}_A(\phi) + \epsilon_{\text{TV}}.
\end{aligned}
\end{equation}
Therefore, enforcing $\hat{R}_A(\phi^*) \leq \alpha - \epsilon_{\text{TV}}$ during MHT-ERM guarantees $R_A(\phi^*) \leq \alpha$ at deployment.}}

% Note that when $\Phi^* = \emptyset$, we set $\Phi^* = \{(0,1)\}$ by convention, which defers all queries to human experts and trivially ensures $R_A(0,1) \leq \alpha$. This completes the proof.
% \qed
% \newpage
\section{A representative example from the TeleQnA dataset} \label{appendix:dataset_example}

The following example illustrates a representative question from the TeleQnA dataset:
\begin{tcolorbox}[colback=purple!5!white,colframe=purple!75!black,title=An Example in TeleQnA Dataset]
{Question:} What is a potential drawback of data fusion-based cooperative wideband sensing techniques?

\begin{itemize}
\item Option 1: Reliable wideband sensing in each cognitive radio
\item Option 2: Saving the total number of measurements
\item Option 3: Heavy data transmission burden in the common control channels
\item Option 4: Improved energy consumption in cellular networks
\item Option 5: Detecting wideband spectrum independently
\end{itemize}

{Answer:} Option 3: Heavy data transmission burden in the common control channels

{Explanation:} Data fusion-based cooperative techniques for wideband sensing can lead to heavy data transmission burden in the common control channels.

{Category:} Research publications
\end{tcolorbox}

\section{{\color{black}{Additional Experiments}}}\label{additional_experiments}

\subsection{{\color{black}{Robustness to validation set size $|\mathcal{D}^{\text{val}}|$}}}\label{appendix:val_size}

{\color{black}{This part presents an additional experiment to validate the robustness of the proposed framework with respect to the size of the validation set $|\mathcal{D}^{\text{val}}|$, which is used to train the Bayesian logistic regression head in the Bayesian learning for white box models. 

Fig.~\ref{fig:val_size_effect_bayes} presents the misalignment and corresponding cost for the cascading systems with thresholds chosen via C-ERM, MHT-ERM-B, and MHT-ERM under different values of $|\mathcal{D}^{\text{val}}|$ varying from $50$ to $300$. We set the target misalignment risk in (\ref{constraint}) to $\alpha = 0.3$ (dashed line) and target reliability in (\ref{FWER_goal}) to $1 - \delta = 0.95$. The results are averaged over $200$ independent experiments (shaded bar on plots shows one standard deviation on both sides).

The results confirm that MHT-ERM and MHT-ERM-B are insensitive to the values of $|\mathcal{D}^{\text{val}}|$, consistently satisfying the alignment constraint across all tested values with only marginal variation in cost. This robustness is consistent with the theoretical guarantee in (\ref{FWER_goal}), which holds irrespective of the quality of the confidence and epistemic uncertainty scores. In contrast, C-ERM exhibits slightly higher sensitivity, since it relies directly on the empirical risk $\hat{R}_A(\phi)$ without any statistical guarantee. As $|\mathcal{D}^{\text{val}}|$ increases and the score quality improves, the misalignment of C-ERM stabilizes near $\alpha = 0.3$.}}

\begin{figure}[htpb]
    \centering
    \includegraphics[width=9cm]{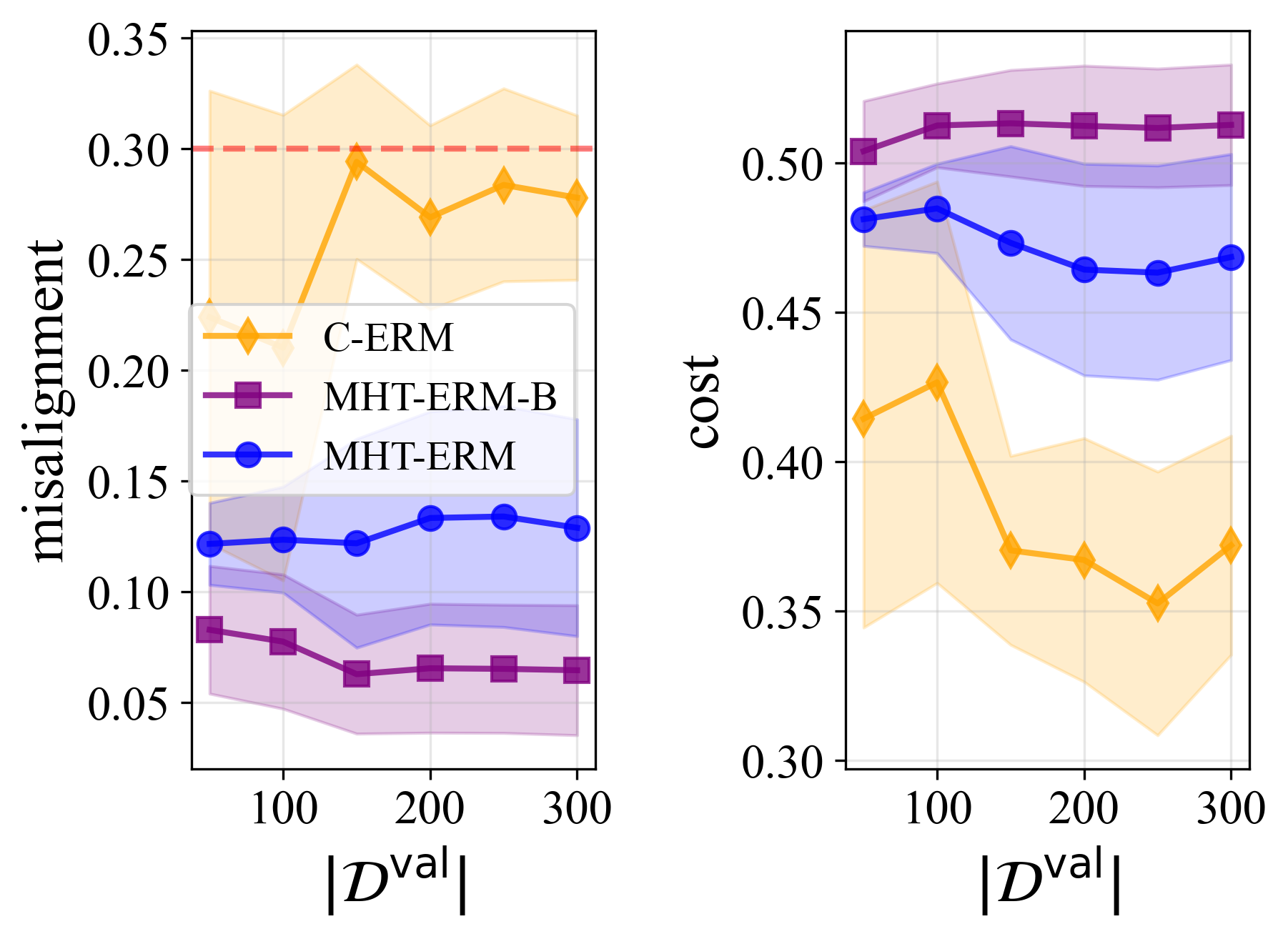}
    \caption{Misalignment and corresponding cost for the cascading systems with thresholds chosen via C-ERM, MHT-ERM-B, and MHT-ERM under different values of $|\mathcal{D}^{\text{val}}|$. We set the target misalignment risk in (9b) to $\alpha = 0.3$ (dashed line) and target reliability in (10) to $1 - \delta = 0.95$. The results are averaged over 200 independent experiments (shaded bar on plots shows one standard deviation on both sides).}
    \label{fig:val_size_effect_bayes}
\end{figure}

\subsection{{\color{black}{Experiments on ORAN-Bench-13K}}}\label{appendix:oran}

{\color{black}{This part presents an additional experiment on ORAN-Bench-13K \cite{gajjar2025oranbench, saenko2026twinpass}, an open-source multiple-choice QA benchmark comprising 13,952 questions derived from 116 O-RAN specification documents, to validate the generalizability of the proposed MHT-ERM framework across different telecom-domain benchmarks. We evaluate the conventional edge-cloud deployment scenario with Qwen2-1.5B-instruct as the edge model and Qwen2-7B-instruct as the cloud model, using the white-box Bayesian learning method for epistemic uncertainty and confidence score evaluation. All other experimental settings follow those described in Sec.~\ref{simulation_setting}.

Fig.~\ref{fig:boxplot_comparison_oran} presents the misalignment and corresponding cost for edge-only, cloud-only, and human-only schemes, as well as for the cascading systems designed via C-ERM, MHT-ERM-B, and MHT-ERM. We set the target misalignment risk in (\ref{constraint}) to $\alpha = 0.15$ (dashed line) and the target reliability in (\ref{FWER_goal}) to $1 - \delta = 0.95$. The colored horizontal lines mark the $1 - \delta = 0.95$-quantile values of the misalignment rate. Maximal values in misalignment and cost performance are reported within the $1.5$ interquartile range (IQR)~\cite{box_plot} across $200$ independent experiments.

The findings are consistent with those on TeleQnA. The proposed MHT-ERM satisfies the alignment constraint while achieving a lower cost than MHT-ERM-B, demonstrating the benefit of the sequential testing strategy. In contrast, C-ERM fails to satisfy the constraint. These results confirm the generalizability of the proposed framework across different telecom-domain benchmarks.}}

\begin{figure}[htpb]
    \centering
    \includegraphics[width=9cm]{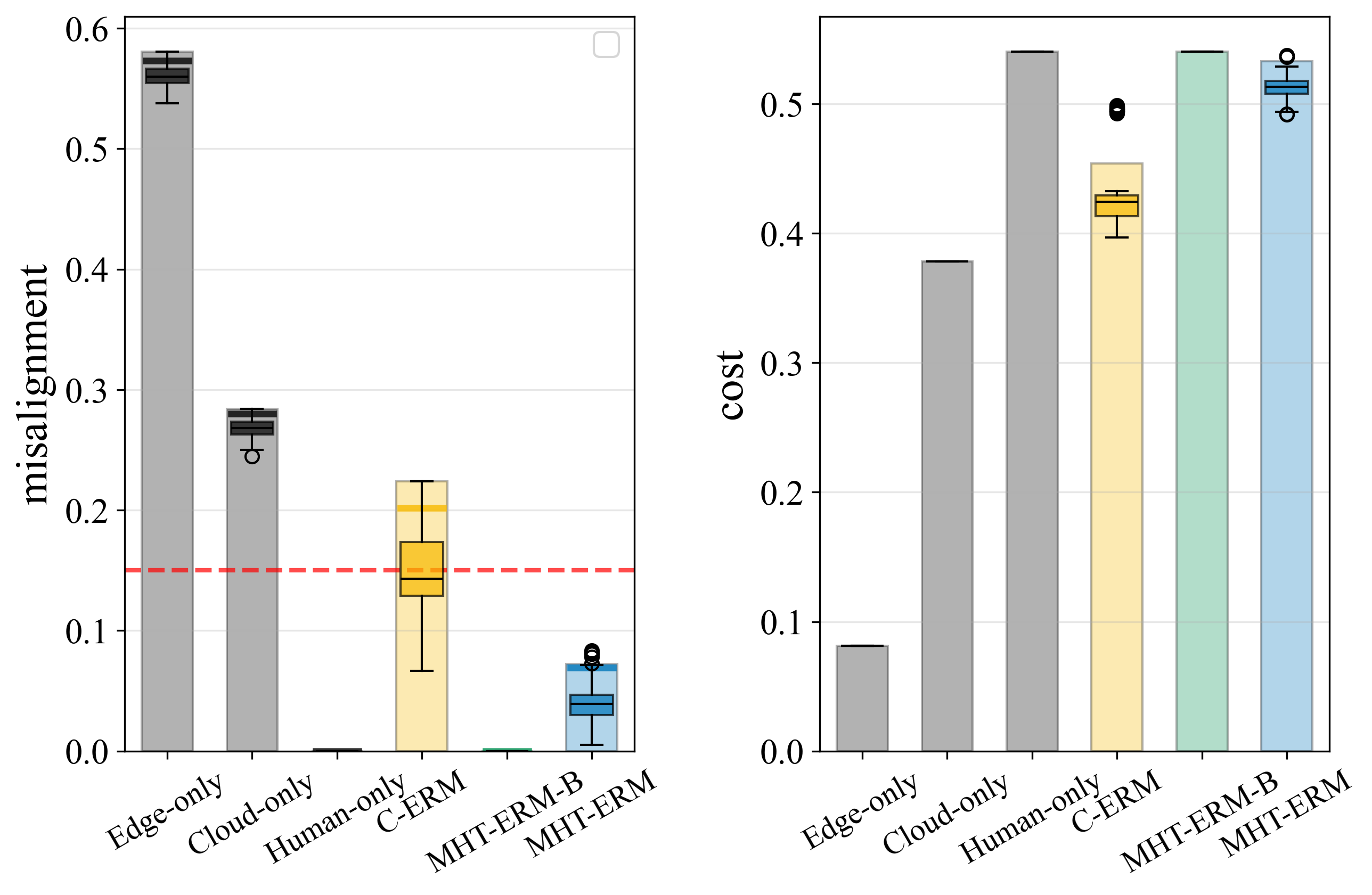}
    \caption{Misalignment and corresponding cost for edge-only, cloud-only, and human-only schemes, as well as for the cascading systems designed via C-ERM, MHT-ERM-B, and MHT-ERM on ORAN-Bench-13K dataset. We set the target misalignment risk in (9b) to $\alpha = 0.15$ (dashed line) and the target reliability in (10) to $1 - \delta = 0.95$. Maximal values are reported within the $1.5$ IQR across 200 independent experiments.}
    \label{fig:boxplot_comparison_oran}
\end{figure}

\end{document}